\documentclass{LMCS}

\def\doi{8(3:23)2012}
\lmcsheading%
{\doi}
{1--20}
{}
{}
{Dec.~15, 2010}
{Sep.~26, 2012}
{}
 
\usepackage[utf8]{inputenc}

\usepackage{enumerate}
\usepackage{hyperref}
\usepackage{xspace}
\usepackage{amsmath, amssymb, latexsym, amsfonts}
\usepackage{tikz}
\usetikzlibrary{automata,arrows,calc,decorations.markings}

\usepackage{macros}
\usepackage{mathtools}
\usepackage{ifthen}

\newenvironment{definition}{\begin{defi}}{\end{defi}}
\newenvironment{theorem}{\begin{thm}}{\end{thm}}

\newenvironment{proposition}{\begin{prop}}{\end{prop}}
\newenvironment{remark}{\begin{rem}}{\end{rem}}
\newenvironment{corollary}{\begin{cor}}{\end{cor}}

\newenvironment{claim}{\medskip\noindent\textit{Claim.\ }}{\medskip}
\begin{document}

\title[Reachability Analysis of Communicating Pushdown Systems]{Reachability Analysis of Communicating Pushdown Systems\rsuper*}

\author[A.~Heußner]{Alexander Heußner}	%required
\address{LaBRI, Universit{\'e} de Bordeaux, CNRS -- France}	%required
\email{\{alexander.heussner, jerome.leroux, anca, gregoire.sutre\}@labri.fr}  %optional

\author[J.~Leroux]{Jérôme Leroux}	%optional
%\address{LaBRI, Universit{\'e} de Bordeaux, CNRS -- France}	%optional
%\email{author2@email2; ditto for email addresses}  %optional
%\thanks{thanks 2, optional.}	%optional

\author[A.~Muscholl]{Anca Muscholl}	%optional
%\address{LaBRI, Universit{\'e} de Bordeaux, CNRS -- France}	%optional
%\email{author3@email3}  %optional
%\thanks{thanks 3, optional.}	%optional

\author[G.~Sutre]{Grégoire Sutre}	%optional
%\address{LaBRI, Universit{\'e} de Bordeaux, CNRS -- France}	%optional
%\email{author3@email3}  %optional
%\thanks{thanks 3, optional.}	%optional

%% etc.

%% required for running head on odd and even pages, use suitable
%% abbreviations in case of long titles and many authors:

%% mandatory lists of keywords and classifications:
\titlecomment{{\lsuper*}An extended abstract of this paper appeared in FoSSaCS'10.}
%%%%%%%%%%%%%%%%%%%%%%%%%%%%%%%%%%%%%%%%%%%%%%%%%%%%%%%%%%%%%%%%%%%%%%%%%%%

%% the abstract has to PRECEED the command \maketitle:
%% be sure not to issue the \maketitle command twice!

\begin{abstract}
  \noindent The reachability analysis of recursive programs that
  communicate asynchro\-nously over reliable \fifo channels calls for
  restrictions to ensure decidability. Our first result characterizes
  communication topologies with a decidable reachability problem
  restricted to eager runs (i.e., runs where messages are either received
  immediately after being sent, or never received). 
 The problem is
  \dexptime-complete in the decidable case.  The second result is a
  doubly exponential time algorithm for bounded context analysis in
  this setting, together with a matching lower bound. Both results
  extend and improve previous work from~\cite{latorre-s-2008-299-a}.
\end{abstract}

\maketitle

\section*{Introduction}%

Checking safety properties for distributed programs like client/server
environments, peer-to-peer applications, or asynchronous programs on
multi-core processors is a standard task in
verification. However, it is well established that the automatic
analysis of distributed programs is a quite challenging
objective. %

A basic feature of the programs used in the applications mentioned above
is that they need to exchange information asynchronously, over
point-to-point channels that are unbounded and reliable. Such
information is used for instance to perform function calls on remote
processes. This amounts to considering a model that combines recursion with
asynchronous communication. Such a combined model is similar in spirit
to, e.g., process rewrite systems~\cite{mayr}, that mix
recursion and Petri nets. We denote the combination of
recursion and asynchronous communication as~\emph{Recursive
  Communicating ProcesseS} (\rqcp for short) here. The model has been
recently studied by La\,Torre, Madhusudan, and
Parlato~\cite{latorre-s-2008-299-a}, who were mainly interested in
applying bounded context analysis to this setting.

Since \rqcp subsume the well-studied class of communicating
finite-state machines~\cite{brand-d-1983-323-a}, reachability is
already undecidable without recursion. Moreover, it is well-known that
reachability for pushdown systems that synchronize by rendezvous
 is undecidable as well~\cite{ramalingam-g-2000-416-a}.
Therefore, our main motivation was to separate these two
sources of undecidability. We consider here behavioral restrictions
for which reachability for communicating \emph{finite-state} machines
is decidable, and then look under which conditions
recursion can be added to the model.

The reachability question for communicating finite-state machines can
be tackled in three different ways, either by restricting
the communication topology, or by assuming that channels are lossy, or by
considering only executions on channels of fixed size. In general,
the last two approaches provide approximated solutions to the
reachability problem. On the positive side, the last idea yields exact
solutions in some special cases, either for certain restricted
topologies (e.g., acyclic ones) or
under certain behavioral restrictions on the communication (e.g., mutex
communication, see below).

As already mentioned, our starting point is the work of La\,Torre et
al.~\cite{latorre-s-2008-299-a}. They introduced a syntactic
restriction on the combined use of channels and pushdowns, that prevents
the synchronization of pushdowns leading to an
undecidable reachability question. An \rqcp is called \emph{well-queueing}
in~\cite{latorre-s-2008-299-a} if pushdown processes can only read
messages when their stack is empty (they can send messages without
any restriction). Well-queueing expresses an event-based programming
paradigm: tasks are executed by threads without interrupt, i.e., a
thread accepts the next task only after it finished the current
one. One of the results of \cite{latorre-s-2008-299-a} is that
well-queueing \rqcp have a decidable reachability problem if and only
if the topology is a directed forest; in the decidable case, they
provide a doubly exponential algorithm by a reduction to bounded-phase
multi-stack pushdown systems~\cite{latorre-s-2007-161-a}.

We extend the results
of~\cite{latorre-s-2008-299-a} in several directions. First, we add a
dual notion to well-queueing: a pushdown process can send messages
only with empty stack (but can read messages without
restriction). This dual notion arises naturally if one wants to model
interrupts: a server might need to accept tasks from high priority
clients independently of the status of the running task. We use these two
restrictions  by fixing the type of each communication channel, to be either
well-queueing or the dual notion. A communication topology, together
with channel types, is called a \emph{typed topology}.

We give in Section~\ref{sect:converge} a precise characterization of those
typed topologies for which the \rqcp model has a
decidable reachability problem over so-called
\emph{eager} runs. A run is eager if the
sending of a message is immediately followed by its
reception (if any). This notion is closely related to 
bounded communication~\cite{lohrey-m-2004-160-a}. 
Communicating finite-state machines with existential channel bounds,
i.e., where each run \emph{can} be
reordered into a run over bounded channels, are a
well-studied model enjoying good expressiveness and
decidability properties~\cite{genest-b-2006-920-a}\footnote{Machines with the
property that  each run can be reordered into an eager
one, are a special instance of existentially 1-bounded
machines. Eagerness is related 
to a \emph{global} channel bound~\cite{lohrey-m-2004-160-a}.}.
Here, we
simply use  eager runs in order to rule out undecidability
due to unbounded channels, since reachability for finite-state
communicating machines over eager runs is decidable. We show that reachability
of \rqcp over eager runs is \dexptime-complete in the
decidable case. Our result generalizes and improves the
doubly exponential time decision procedure of
\cite{latorre-s-2008-299-a}, which holds for topologies without undirected
cycles (called \emph{polyforests}).

The restriction to eager runs appears to be strong at a first glance.
However, we show in Section~\ref{sect:mutex} that it arises rather
naturally, by imposing a behavioral restriction on the communication:
the \emph{mutex} restriction requires that in every reachable
configuration there is no more than one non-empty channel per cycle of
the network. In particular, \rqcp over polyforest architectures are
mutex. Mutex can also be  seen as a generalization of the half-duplex
restriction studied in \cite{cece-g-1997-304-a}.

La\,Torre et al.~propose in~\cite{latorre-s-2008-299-a}  a second
approach to solve the reachability problem for \rqcp, inspired by successful
work on reachability with bounded contexts in the
verification of concurrent Boolean programs~\cite{qadeer-s-2005-93-a}.
They show that bounded-context reachability for
well-queueing \rqcp is decidable in time doubly exponential in the
number of contexts. Again, this result is obtained by a reduction to
bounded-phase multi-stack pushdown
systems~\cite{latorre-s-2007-161-a}. Our result in Section~\ref{sect:bounded}
extends the bounded-context result of~\cite{latorre-s-2007-161-a} to
\rqcp that allow for the two dual notions of well-queueing.
Moreover, our algorithm is direct and simpler than the one
involving bounded-phase multi-stack pushdown
systems. We also provide a matching lower bound for the complexity.

\medskip
\noindent{\it Related work.\ }
In the context of multi-thread programming, other notions of
synchronization between pushdowns arise naturally. Earlier
publications considered synchronization via shared memory, such as
local/global memory in
\cite{bouajjani-a-2005-348-a,bouajjani-a-2005-437-a} or bags in
\cite{sen-k-2006-300-a,jhala-r-2007-339-a}. The paper
\cite{bouajjani-a-2005-348-a} showed that bounded-context reachability
can be solved in exponential time, whereas \cite{sen-k-2006-300-a}
provided an exponential space lower bound for reachability with atomic
methods (without
context bounds).  Also, synchronization in the form of state
observation was considered in~\cite{atig-m-2008-356-a}. The latter
model was shown to be decidable only for acyclic architectures, and is
strongly related to lossy
systems~\cite{abdulla-p-1996-91-a,finkel-a-2001-63-a}.  For the shared
memory model, \cite{kidd10} shows how to reduce concurrent pushdowns
to a single pushdown, assuming a priority preemptive scheduling policy.  Lately,
\cite{seth-a-2010-615-a,atig-m-2010-117-b} proposed a general strategy
to reduce bounded-phase reachability questions on different
multi-stack pushdown automata models to a single stack. This is close
in spirit to our proof technique in Section~\ref{sect:converge},
although we do not rely on a phase-bounded model for our first
result.

\section{Recursive Communicating Processes}

Given a set $P$ and a $P$-indexed family of sets $(S^p)_{p \in P}$, we 
write elements of the Cartesian product $\prod_{p \in P} S^p$ in bold face.  For
any $\vec{s}$ in $\prod_{p \in P} S^p$ and any $p \in P$, we let $s^p \in S^p$
denote the $p$-component of  $\vec{s}$.  Moreover, we identify $\vec{s}$
with the indexed family of elements $(s^p)_{p \in P}$.

An \emph{alphabet} is any finite set of \emph{letters}.  Given an
alphabet $\Sigma$, we write $\Sigma^*$ for the set of all \emph{finite words}
(\emph{words} for short) over $\Sigma$, and we let $\varepsilon$ denote the
\emph{empty word}.  %

A \emph{labeled transition system} (LTS for short)
$\Lts=\tuple{S,s_\init,A,\rightarrow}$ is given by a set
of~\emph{states} $S$, an initial state $s_\init$, an 
\emph{action} alphabet $A$, and a (labeled)~\emph{transition relation}
$\rightarrow$, which is a subset of $S\times A\times S$. For
simplicity, we usually write $s\act{a}s'$ in place of $(s,a,s')\in\
\rightarrow$. The \emph{size} of $\Lts$ is defined by
$\size{\Lts} = \size{S}^2 \cdot \size{A}$
when $S$ is finite.

Throughout the paper we use standard complexity classes such as
polynomial space (\pspace), deterministic exponential time (\dexptime), and
deterministic doubly-exponential time (\twodexptime).  For detailed definitions
the reader is referred to, e.g., \cite{papadimitriou-c-1994--a}.

\subsection{Communication Topologies}

In this paper, we consider processes from a finite set $P$, that
communicate over point-to-point, error-free \fifo channels from a set
$C$. They exchange messages over a given  topology, which
is simply a directed graph whose vertices are processes and whose
edges represent channels:

\begin{definition}
  A \emph{topology} $\Topo$  is a tuple $\tuple{P,C,\src,\tgt}$ where $P$
  is a finite set of \emph{processes}, and $C$ is a finite set
  of point-to-point \emph{channels} equipped with two functions
  $\src,\tgt:C\rightarrow P$ that map every channel $c\in C$ to a
  \emph{source} $\src(c) \in P$ and a \emph{destination}
  $\tgt(c) \in P$, such that $\src(c) \neq \tgt(c)$.
\end{definition}

The \emph{size} of $\Topo$ is defined by $\size{\Topo} = \size{P}+\size{C}$.
For each channel $c \in C$, we write $\dchannel{}{c}{}$ for the binary
relation on the set of processes $P$ defined by $p\dchannel{}{c}{}q$
if $p=\src(c)$ and $q=\tgt(c)$. We also use the undirected
binary relation $\uchannel{}{c}{}$, defined by $p\uchannel{}{c}{}q$ if
$p\dchannel{}{c}{}q$ or $q\dchannel{}{c}{}p$.

An \emph{undirected path} in $\Topo$ is an alternating sequence
$(p_0,c_1,p_1,\ldots,c_n,p_n)$, of processes $p_i \in P$ and
channels $c_i \in C$,
such that $p_{i-1}\uchannel{}{c_i}{}p_i$ for all $i$.
Moreover, the undirected path is called \emph{simple} if the processes
$p_0, \ldots, p_n$ are distinct.
A \emph{simple undirected cycle} in $\Topo$ is an undirected path
$(p_0,c_1,p_1,\ldots,c_n,p_n)$ with $p_0 = p_n$
such that $p_1, \ldots, p_n$ are distinct, and
$c_1, \ldots, c_n$ are distinct. 
The topology $\Topo$ is called \emph{polyforest} if it contains no
simple undirected cycle.

\subsection{Communicating  Processes}

Consider a topology $\Topo = \tuple{P,C,\src,\tgt}$.
Given a message alphabet $M$, we denote by $\Com^p(\Topo,M)$ the set
of \emph{possible communication actions} of a process $p \in P$,
defined by $\Com^p(\Topo,M)=\{c!m \mid c\in C, \src(c)=p, m\in
M\}\cup\{c?m \mid c\in C, \tgt(c)=p, m\in M\}$.  As usual, $c!m$
denotes sending message $m$ into channel $c$, whereas $c?m$
denotes receiving message $m$ from channel $c$.
Note that $\Com^p(\Topo,M)$ and $\Com^q(\Topo,M)$ are disjoint when
$p$ and $q$ are distinct processes.

\begin{definition}
  A \emph{system of communicating processes (\qcp for short)}
  $\Qcp=\tuple{\Topo,M,(\Lts^p)_{p\in P}}$ is given by a topology
  $\Topo$, a \emph{message} alphabet $M$, and, for each process
  $p\in P$, an LTS $\Lts^p=\tuple{S^p,s_\init^p,A^p,\rightarrow_p}$
  such that:
  \begin{iteMize}{$\bullet$}
  \item the action alphabets $A^p$, $p \in P$, are pairwise disjoint, and
  \item $A^p_\actcom = A^p \cap (C\times\{!,?\}\times M)$ is contained in
    $\Com^p(\Topo,M)$ for each $p \in P$.
  \end{iteMize}
\end{definition}

\noindent Actions in $A^p_\actcom$ are called \emph{communication actions} of
$p$, whereas $A_\actloc^p=A^p\moins A_\actcom^p$ is the set of
\emph{local actions}. States $s^p\in S^p$ are called
\emph{local states} of $p$. We write $\vec{S}=\prod_{p\in P}S^p$ for the
set of \emph{global states}. %
Note that the sets $S^p$, and
hence $\vec{S}$, may be infinite.  Indeed, the local transition
systems $\Lts^p$ could be, for example, counter
or pushdown systems.  When $\vec{S}$
is finite, $\Qcp$ is called a \emph{finite} \qcp,
and its \emph{size} is defined by
$\size{\Qcp} = \size{\Topo} + \size{M} + \sum_{p \in P} \size{\Lts^p}$.

\smallskip

As usual, the semantics of \qcp is defined in terms of a global LTS
  $\tuple{X,x_\init,A,\rightarrow}$, where $X = \vec{S} \times
  (M^*)^C$ is the set of \emph{configurations}, $x_\init = (\vec{s_\init},
   (\varepsilon)_{c \in C})$ is the initial configuration,
  $A = \bigcup_{p \in P} A^p$ is the set of actions, and
  $\rightarrow \mathop{\subseteq} X \times A \times X$ is
  the transition relation with $(\vec{s_1}, \vec{w_1}) \act{a} (\vec{s_2},
  \vec{w_2})$, where $a \in A^p$, if 
  the following conditions are satisfied:
  \begin{enumerate}[(i)]
  \item $s_1^p \xrightarrow{a}_p s_2^p$ and $s_1^q = s_2^q$ for all $q
    \in P$ with $q \neq p$,
  \item if $a \in A^p_\actloc$ then $\vec{w_1} = \vec{w_2}$,
  \item if $a = c!m$ then $w_2^c = w_1^c \cdot m$ and $w_2^d = w_1^d$ for all $d \in C$ with $d \neq c$,
  \item if $a = c?m$ then $m \cdot w_2^c = w_1^c$ and $w_2^d = w_1^d$ for all $d \in C$ with $d \neq c$.
  \end{enumerate}
Given a process $p \in P$, we call~\emph{move} of $p$ any transition
$x_1 \act{a} x_2$ with $a \in A^p$. A move is local if $a \in
A_\actloc^p$.

\medskip

A \emph{run} in the LTS $\Qcp$ is a
finite, alternating sequence $\rho=(x_0,a_1,x_1,\ldots,a_n,x_n)$ of
configurations $x_i\in X$ and actions $a_i\in A$ satisfying
$x_{i-1} \act{a_i} x_i$ for all $i$. We say
that $\rho$ is a run from $x_0$ to $x_n$. The
\emph{length} of $\rho$ is $n$, and is denoted by $|\rho|$.
A run of length zero consists of a single configuration.
The \emph{trace} of a run
    $\rho = (x_0,a_1,x_1,\ldots,a_n,x_n)$ is the sequence of actions
    $\trace(\rho) = a_1 \cdots a_n$.
A pair of send/receive
actions $a_i=c!m, a_j=c?m$ is called~\emph{matching} in
$\rho$ if $i<j$ and the number of receives on $c$ within
$a_i \cdots a_j$ equals the length of $c$ in $x_i$.
If $\rho,\rho'$ are two runs such that the last configuration of $\rho$
is equal to the first configuration of $\rho'$, then we write
$\rho\cdot\rho'$ for their concatenation.

We
define the \emph{order-equivalence} relation $\sim$ over runs as the
finest congruence such that $(x_0,a,x_1,b,x_2) \sim
(x_0,b,x'_1,a,x_2)$ whenever $a,b$  are actions on different processes.
Informally, $\rho \sim \rho'$ if they can be
transformed one into the other by iteratively commuting adjacent
transitions that (i) are \emph{not} located on the same
process, and (ii) do \emph{not} form a matching send/receive pair.
The following is easy to check:

\begin{fact}
  If $\rho,\rho'$ are order-equivalent runs of a \qcp, then
  they start in the same configuration and end in the same configuration.
\end{fact}

\medskip

A configuration $x\in X$ is \emph{reachable} in a \qcp $\Qcp$ if there
exists a run of $\Qcp$ from the initial configuration $x_\init$ to
$x$.  We define the \emph{reachability set} of $\Qcp$ as
$\Reach(\Qcp)=\{x \in X \mid x \text{ is reachable in } \Qcp\}$.

\smallskip

The~\emph{state reachability problem} for \qcp asks, for a given \qcp
$\Qcp$ and a global state $\vec{s} \in \vec{S}$, whether 
$\Reach(\Qcp)$ intersects
$\{\vec{s}\} \times (M^*)^C$.  It is well-known that this problem is
undecidable for finite \qcp, even if we restrict the
topology %
to two
processes connected by two channels~\cite{brand-d-1983-323-a}.

The undecidability of the state reachability problem for \qcp is based on
the fact that one cannot control how ``fast'' messages are received.
A simple idea that rules out such behaviors is to consider only runs
where the reception is immediate (if it exists):

\begin{definition}
  A run $\rho=(x_0,a_1,x_1,\ldots,a_n,x_n)$ is \emph{eager} if for all
  $1 \leq i \leq n$, if $a_i$ is a receive action then $i > 1$ and
  $a_{i-1}$ is its matching send action. %
    \label{def:bounded_run}
\end{definition}
Thus, each send action along an eager run is either immediately
followed by its matching receive, or it is never matched.  In
the latter case, all later sends into the channel are never received,
and we say that the channel is in its ``growing phase''.  In the
former case, the adjacent matched send/receive actions act like a
rendezvous synchronization between the two processes.  Formally, given
a channel $c \in C$, we call \emph{rendezvous on $c$} any run (of
length~2) $\rho = (x, c!m, x', c?m, x'')$ such that $x= (\vec{s},
\vec{w})$ with $w^c = \varepsilon$. The rendezvous~\emph{involves
  process $p$} if $p \in \{\src(c),\tgt(c)\}$.

We introduce now the ``eager'' variants of the reachability notions
presented previously.  A configuration $x\in X$ is
\emph{eager-reachable} in a \qcp $\Qcp$ if there exists an eager run
from the initial configuration $x_\init$ to $x$.  The
\emph{eager-reachability set} of $\Qcp$ is the set
$\Reach_\eager(\Qcp)$ of eager-reachable configurations.  We say that
a \qcp $\Qcp$ is~\emph{eager} when $\Reach_\eager(\Qcp) =
\Reach(\Qcp)$. In the next section, we show how eager \qcp 
occur under some natural (and decidable) restrictions on cyclic
communication. The simplest example arises over polyforest
topologies. 
The \emph{state eager-reachability problem} for \qcp
asks, for a \qcp $\Qcp$ and a global state $\vec{s}\in
\vec{S}$, whether $\Reach_\eager(\Qcp)$ intersects $\{\vec{s}\} \times
(M^*)^C$.  It is readily seen that this problem is decidable for
\emph{finite} \qcp in \pspace.

Eager runs, modulo the fact that Definition~\ref{def:bounded_run}
allows for runs which end in a sequence of (unmatched) send actions,
are closely related to the notion of globally 1-bounded runs. Eager
\qcp subsume existentially globally 1-bounded communicating
machines~\cite{lohrey-m-2004-160-a,genest-b-2007-1-a}. However, as we
will see in Section~\ref{sect:mutex}, it is undecidable whether a
finite \qcp is eager (in contrast, one can decide whether a finite,
deadlock-free communicating machine is existentially globally 1-bounded~\cite{genest-b-2007-1-a}). On the
positive side, Section~\ref{sect:mutex} shows a decidable subclass of
finite, eager \qcp.

\subsection{Recursive Communicating Processes}\label{ssec:rqcp}

In the following we introduce \rqcp together with a symmetric version
of the ``well-queueing'' restriction used in
\cite{latorre-s-2008-299-a}.
Informally, \rqcp (recursive \qcp) are  \qcp where each local
transition system is a pushdown system.

A well-queueing %
\rqcp in
\cite{latorre-s-2008-299-a} is one where a process can only receive
when its stack is empty. Here, we dualize this concept by also allowing
channels where the sender (but not the receiver) must have an empty
stack. Well-queueing was motivated
in~\cite{latorre-s-2008-299-a} by the case where recursive processes
need to finish their tasks before accepting new ones. Adding the dual notion
of well-queueing is interesting when modeling interrupts: a
recursive process may have to interrupt its current task to treat one
with a higher priority, hence, it has to preserve its current state on
the stack to return later.

\begin{definition}
  A \emph{typed} topology $\tuple{\Topo,\tau}$ consists of a topology
  $\Topo$, together with a \emph{type} $\tau \subseteq P \times C$, such that $(p,c)
  \in\tau$ implies $p \in \{\src(c),\tgt(c)\}$.
  \label{def:oriented}
\end{definition}

Given a process $p \in P$ and a channel $c \in C$, we call
$p$~\emph{restricted} on $c$ if $(p,c) \in \tau$ (and
\emph{unrestricted} otherwise). Informally, a communicating pushdown
process $p$ as defined below will be restricted on $c$ if $p$'s stack
must be empty when communicating over channel $c$.

\begin{definition} 
  A \emph{pushdown system} $\Pdp = \tuple{Z, z_\init, A,
    A_\varepsilon,\Gamma, \Delta}$ is given by a finite set $Z$ of
  \emph{control states}, an \emph{initial control state} $z_\init \in
  Z$, an alphabet $A$ of \emph{actions},
  a subset $A_\varepsilon \subseteq A$,
  a stack alphabet $\Gamma$, and
  a transition relation $\Delta \subseteq Z \times A \times Z$, such that
  $A$ contains the set $A_\stack = \{\push(\g),\pop(\g) \mid \g \in \G\}$
  of \emph{stack actions}.
\end{definition}
We define the \emph{size} of $\Pdp$ by $\size{\Pdp} = \size{Z}^2 \cdot
\size{A}$.
Actions in $A_\varepsilon \subseteq A \setminus A_\stack$ are tests for
empty stack.
Naturally, for a pushdown system embedded in a \qcp, the set of actions
$A \setminus A_\stack$ may contain communication (and local) actions.
Depending on the typed topology, some communication actions
may require an empty stack.  This will be enforced by putting these
communication actions in the set $A_\varepsilon$.

According to the informal description given above, we define now the
semantics of pushdown processes.
The semantics of $\Pdp = \tuple{Z, z_\init, A, A_\varepsilon,\Gamma, \Delta}$
is the LTS $\tuple{S,s_\init,A,\rightarrow}$ with set of states
$S=Z \times \Gamma^*$, initial state $s_\init = (z_\init, \varepsilon)$,
and (labeled) transition relation $\rightarrow$ defined as expected:
stack actions $\push(\g)$ and $\pop(\g)$ behave as usual ($\pop(\g)$ blocks
if the top of the stack is not $\g$), actions from $A \setminus A_\stack$
do not change the stack, and actions in $A_\varepsilon$ are possible only if
the stack is empty.

\begin{definition}
  A \emph{recursive \qcp} (\rqcp for short)
  $\Rqcp=\tuple{\Topo,\tau,M,(\Pdp^p)_{p\in P}}$ is given by a typed
  topology $\tuple{\Topo,\tau}$, a~\emph{message} alphabet $M$, and,
  for each  process $p\in P$,  a
  pushdown system
  $\Pdp^p=(Z^p,z_\init^p,A^p,A^p_\varepsilon,\Gamma^p,\Delta^p)$ such that:
  \begin{iteMize}{$\bullet$}
  \item the action alphabets $A^p$, for $p \in P$, are pairwise disjoint,
  \item $A^p_\actcom = A^p \cap (C\times\{!,?\}\times M)$ is contained in
    $\Com^p(\Topo,M)$ for each $p \in P$, and
  \item $A^p_\varepsilon \supseteq \{c!m \in A^p_\actcom \mid (p, c) \in \tau\} \cup
    \{c?m \in A^p_\actcom \mid (p, c) \in \tau\}$ for each $p \in P$.
  \end{iteMize}
\end{definition}
We associate with $\Rqcp$ the \qcp %
$\tuple{\Topo,M,(\Lts^p)_{p\in P}}$
where, for each $p \in P$, the LTS $\Lts^p$ is the semantics of the
pushdown system $\Pdp^p$.
The \emph{size} of $\Rqcp$ is defined by
$\size{\Rqcp} = \size{\Topo} + \size{M} + \sum_{p \in P} \size{\Pdp^p}$.

\smallskip

We write $\vec{Z}=\prod_{p\in P}Z^p$ for the set
\emph{global control states}.
Abusing notation, a global state $\vec{s}$ of $\Rqcp$ will also be written
$\vec{s} = (\vec{z}, \vec{u})$ where $s^p = (z^p, u^p)$ for each $p \in P$.
The~\emph{state reachability problem} for \rqcp
asks, for a given \rqcp
$\Rqcp$ and a global control state $\vec{z} \in \vec{Z}$, whether
$\Reach(\Rqcp)$
intersects
$\{\vec{z}\} \times (\prod_{p \in P} (\Gamma^p)^*) \times (M^*)^C$.
The \emph{state eager-reachability problem} for \rqcp is defined similarly,
using $\Reach_\eager(\Rqcp)$ instead of $\Reach(\Rqcp)$.

\section{Topologies with Decidable State Reachability}
   \label{sect:converge}

   Several factors lead to the undecidability of the state
   reachability problem for \rqcp.  In particular, the model is
   already undecidable without any pushdown.  Our goal in this section
   is a decidability condition that concerns the interplay between
   pushdowns and communication, assuming that the communication is
   \emph{not} the reason for undecidability.  For this reason, we consider a
   restricted version of the state reachability problem, namely the
   one on eager runs.

\begin{definition}\label{def:conv}
  A typed topology $\tuple{\Topo,\tau}$ is called
  \emphconverging if it contains a simple undirected path
  $(p_0,c_1,p_1,\ldots,c_n,p_n)$, with $n \geq 1$, such that
  $p_0$ is unrestricted on $c_1$ and $p_n$ is unrestricted on $c_n$.
\end{definition}

Notice that non-\convergence
implies that
every channel is either restricted
at the source, or at the destination, or at both ends
(see Figure~\ref{fig:non-confluent-examples}).

\begin{figure}[b]
  \centering
\begin{tikzpicture}[
  node distance=2cm,
  decoration={
    markings,
    mark=at position 0.62 with {\arrow{triangle 60}}
  }]

  \node (p)  {$p$};
  \node (q1)  [below left of=p]  {$q_1$};
  \node (qn)  [below right of=p] {$q_n$};

  \draw      (p) edge [*-o] (q1)
  [decorate] (p) to         (q1)
  ;
  \draw      (p) edge [*-o] (qn)
  [decorate] (p) to         (qn)
  ;

  \draw [dotted, bend angle=30, bend right]
  ([xshift=1mm] q1.north east) to ([xshift=-1mm] qn.north west)
  ;
\end{tikzpicture}
\hspace{0.75cm}
\begin{tikzpicture}[node distance=0.5cm]
  \node (p)  {};
  \node (p') [right of=p, anchor=west]  {\small\em restricted};
  \node (q)  [below of=p] {};
  \node (q') [right of=q, anchor=west] {\small\em unrestricted};

  \draw (p) edge [o-] (p')
        (q) edge [*-] (q')
  ;
\end{tikzpicture}
\hspace{0.75cm}
\begin{tikzpicture}[
  proc/.style = {}
  ]

   \draw (20:1.5)  node[proc] (1) {$p_1$};
   \draw (110:1.5) node[proc] (2) {$p_2$};
   \draw (200:1.5) node[proc] (3) {$p_3$};
   \draw (290:1.5) node[proc] (4) {$p_4$};

   \begin{scope}[
     decoration={
       markings,
       mark=at position 0.54 with {\arrow{triangle 60}}
     }]

     \draw      (1) edge[*-o,bend right=30] (2)
     [decorate] (1) to  [    bend right=30] (2)
     ;
     \draw      (2) edge[*-o,bend right=30] (3)
     [decorate] (2) to  [    bend right=30] (3)
     ;
     \draw      (3) edge[*-o,bend right=30] (4)
     [decorate] (3) to  [    bend right=30] (4)
     ;
     \draw      (4) edge[*-o,bend right=30] (1)
     [decorate] (4) to  [    bend right=30] (1)
     ;
   \end{scope}

   \begin{scope}[
     decoration={
       markings,
       mark=at position 0.46 with {\arrowreversed{triangle 60}}
     }]

     \draw      (1) edge[o-o,bend left=20] (2)
     [decorate] (1) to  [    bend left=20] (2)
     ;
     \draw      (2) edge[*-o,bend left=20] (3)
     [decorate] (2) to  [    bend left=20] (3)
     ;
     \draw      (3) edge[*-o,bend left=20] (4)
     [decorate] (3) to  [    bend left=20] (4)
     ;
     \draw      (4) edge[o-o,bend left=20] (1)
     [decorate] (4) to  [    bend left=20] (1)
   ;
   \end{scope}
\end{tikzpicture}
  \caption{Examples of non-\converging typed topologies}
  \label{fig:non-confluent-examples}
\end{figure}
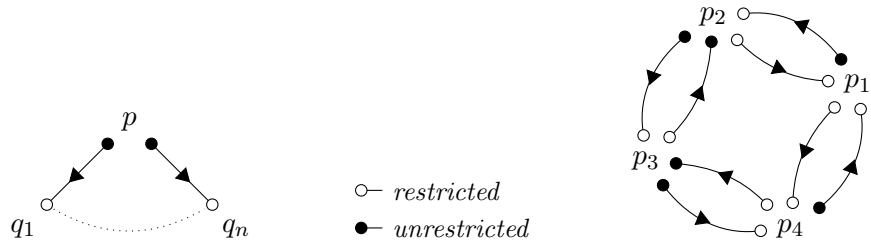

We say that a typed topology $\tuple{\Topo,\tau}$ has a decidable \rqcp state
eager-reachability problem if the latter question is decidable for the class of
\rqcp with typed topology $\tuple{\Topo,\tau}$.  We show in this section that the notion of
\convergence gives a complete characterization of typed topologies
with respect to the decidability of the above problem.

\begin{theorem}\label{thm:conv}
  A typed topology has a decidable \rqcp state
 eager-reachability problem  if and only if it is non-\converging.
  Moreover, the problem is \dexptime-complete in the latter case.
\end{theorem}

The rest of the section is devoted to the proof of this theorem.
We first show the undecidability result in the \converging case.

\begin{proposition}
  Every \converging typed topology has an undecidable \rqcp
  state (eager-) reachability problem. 
\end{proposition}
\begin{proof}
  Consider a typed topology $\tuple{\Topo,\tau}$ that is \converging.  There
  is a simple undirected path $p_0 \uchannel{}{c_1}{} p_1
  \cdots p_{n-1} \uchannel{}{c_n}{} p_n$
  satisfying the conditions of Definition~\ref{def:conv}.
  Since $p_0$ is unrestricted on $c_1$ and $p_n$ is
  unrestricted on $c_n$, both may use their stack while communicating
  over the channels $c_1$ and $c_n$, respectively. 
  Recall that checking non-emptiness of the intersection of two context-free languages is undecidable.
  To prove the lemma, we reduce this problem to the state
  eager-reachability problem for \rqcp with typed topology $\tuple{\Topo,\tau}$.

  Given two context-free languages $K$ and $L$ over the alphabet
  $\{\mathtt{0}, \mathtt{1}\}$, the process $p_0$ guesses a word in
  $K$ while $p_n$ guesses a word in $L$, and both processes check that
  they guessed the same word via synchronizations along the undirected path $p_0
  \uchannel{}{c_1}{} p_1 \cdots p_{n-1} \uchannel{}{c_n}{} p_n$. 
  Intermediate processes $p_1, \ldots, p_{n-1}$ do not use their
  stack, they simply convey the information about the common input
  guessed by $p_0$ and $p_n$.  The labeled transition system
  $\Lts^{p_i}$, $1 \leq i < n$, is depicted below.

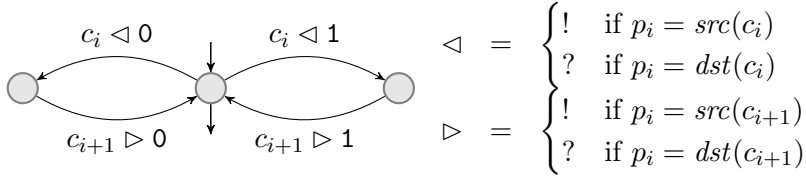
\begin{tikzpicture}[node distance=2.5cm, >=stealth', bend angle=30]
  \tikzstyle{initial}   = [initial   by arrow, initial   text=, initial   above, initial   distance=0.4cm]
  \tikzstyle{accepting} = [accepting by arrow, accepting text=, accepting below, accepting distance=0.4cm]

  \tikzstyle{every state} = [draw=gray, thick, fill=gray!20, minimum size=4mm, inner sep=2pt]

  \node [state]  (q1)  [initial, accepting]         {};
  \node [state]  (q2)  [left of=q1]      {};
  \node [state]  (q3)  [right of=q1]     {};
  \node                [right of=q3, node distance=3cm]     {%
$\begin{array}{rcl}
\lhd & = &
\begin{cases}
  ! & \mbox{if } p_i = \src(c_i)\\
  ? & \mbox{if } p_i = \tgt(c_i)\\
\end{cases}
\\
\rhd & = &
\begin{cases}
  ! & \mbox{if } p_i = \src(c_{i+1})\\
  ? & \mbox{if } p_i = \tgt(c_{i+1})\\
\end{cases}
\end{array}$};

  \path[->]  (q1)  edge [bend right]  node [above]   {$c_i \lhd \mathtt{0}$}     (q2)
             (q2)  edge [bend right]  node [below]   {$c_{i+1} \rhd \mathtt{0}$}     (q1)
             (q1)  edge [bend left]   node [above]   {$c_i \lhd \mathtt{1}$} (q3)
             (q3)  edge [bend left]   node [below]   {$c_{i+1} \rhd \mathtt{1}$} (q1)
  ;
\end{tikzpicture}%

  Similarly, the pushdown systems $\Pdp^{p_0}$ and $\Pdp^{p_n}$ are obtained
  from pushdown automata accepting $K$ and $L$, respectively, by replacing
  tape-reading actions with communications ($c_1 \rhd \mathtt{0}/\mathtt{1}$
  for $p_0$ and $c_n \lhd \mathtt{0}/\mathtt{1}$ for $p_n$).

 Finally, we only need to make sure that the channels are empty at the end.
  As usual, this can be enforced by augmenting $M$ with a new symbol
  $\mathtt{\$}$, and by sending and receiving $\mathtt{\$}$ on each
  channel $c$ at the end of the simulation. 

The construction guarantees that the intersection $K \cap L$ is
non-empty if and only if there is an (eager) run in the \rqcp from the
initial configuration to a global control state where each process is
accepting. 
\end{proof}

We now focus on non-\converging typed topologies.
Let us first prove the \dexptime lower bound of
Theorem~\ref{thm:conv}.

\begin{proposition} \label{prop:non-confluent-lb} The state
 eager-reachability problem for  \rqcp
  with non-\converging typed topology is  \dexptime-hard.
\end{proposition}

\begin{proof}
It is well-known (and probably folklore) that the following problem is
\dexptime-complete: 
given a context-free language $K$ and $n$ regular languages $L_i$,
check the non-emptiness of $K \cap \bigcap_i L_i$.
The hardness follows easily by a reduction from linearly bounded
alternating Turing 
machines.
Actually, a closely related problem is shown to be \dexptime-hard
in~\cite{esparza-j-2003-355-a}, namely the reachability problem for
pushdown systems with
checkpoints.

Notice that the intersection $K \cap \bigcap_i L_i$ can be simulated
on the non-\converging, 
typed topology $\tuple{\Topo,\tau}$ where
$P = \{p, q_1, \ldots, q_n\}$, $C = \{c_1, \ldots, c_n\}$, and,
for each $1 \leq i \leq n$,
$p\dchannel{}{c_i}{} q_i$ with $p$ unrestricted on $c_i$ and
$q_i$ restricted on $c_i$ (see left part of
Figure~\ref{fig:non-confluent-examples}). That is, process $p$
simulates a pushdown automaton accepting the context-free language
$K$, whereas process $q_i$ simulates a finite-state automaton
accepting $L_i$. Communication guarantees that the simulations
use the same input word. As in the previous proposition, one needs to
enforce the emptiness of the channels by using an extra symbol.
\end{proof}

Before considering the upper bound we need to introduce some
vocabulary. Consider a run $\rho = (x_0,a_1,x_1,\ldots,a_n,x_n)$ of an
\rqcp $\Rqcp$.  Given a process $p\in P$, we say that $\rho$ is
\emph{well-formed for $p$} if the projection of $a_1 \cdots
a_n$ on $A^p_\stack$ is a  Dyck word.  This
well-formedness condition merely stipulates that each push action of
$p$ in $\rho$ is matched by a pop action, and vice versa.  We
call $\rho$ \emph{well-formed} if $\rho$ is well-formed for each
process $p \in P$.  For instance, every run that starts and ends with
empty stacks  is well-formed.  A stronger condition is that of
well-bracketing, which requires that push and pop actions for distinct
processes must be nested recursively.
Formally, we say that $\rho = (x_0,a_1,x_1,\ldots,a_n,x_n)$ is
\emph{well-bracketed} if the following two conditions are satisfied:
\begin{enumerate}
\item
  the projection of $a_1 \cdots a_n$ on the disjoint union
  $\bigcup_{p \in P} A^p_\stack$ is a Dyck word, and
\item
  for every process $p$ and every $h < i < j < k$,
  if the pairs $(a_h,a_k)$ and $(a_i,a_j)$ are matching push/pop actions
  of $p$, then the sub-runs $(x_{h-1},a_h,x_h,\ldots,a_i,x_i)$ and
  $(x_{j-1},a_j,x_j,\ldots,a_k,x_k)$ are
  well-formed for all $q \neq p$.
\end{enumerate}
Observe that if $\rho\cdot\rho'$ is defined, then $\rho\cdot\rho'$ is
well-formed (resp.~well-bracketed) if  $\rho$ and $\rho'$ are both
well-formed (resp.~well-bracketed).  Note also that well-formedness is
preserved under order-equivalence: if $\rho$ is well-formed and $\rho
\sim \rho'$ then $\rho'$ is also well-formed.  However,
well-bracketing is not preserved under order-equivalence.

The following proposition provides the main ingredient to show the
$\dexptime$ upper bound of Theorem~\ref{thm:conv}.

\begin{proposition} \label{prop:wf-to-wb-for-confluent} Given an \rqcp
  $\Rqcp$ with non-\converging typed topology, every eager,
  well-formed run in $\Rqcp$ is
  order-equivalent to an eager, well-bracketed run. 
\end{proposition}

\begin{proof}
  By induction on the length of runs.
  The basis is trivial.
  Consider a run $\rho$, of non-zero length, that is both eager and well-formed.
  We assume that $\rho$ starts with a push action (otherwise, the existence of an order-equivalent run that is both eager and well-bracketed immediately follows by induction).
  Let $a=\push(\gamma)$ denote the first action of $\rho$, and let $p$ denote the process with $a \in A^p$.
  Let $\rho'$ denote an order-equivalent eager run obtained from $\rho$ by scheduling the actions of $p$ as early as possible, while maintaining adjacent send/receive pairs.
  It is readily seen that $\rho'$ may be written as:
  $$
  \rho' \ = \ x \longact{\push(\gamma)} x' \cdot \pi_0 \cdot \chi_1 \cdot \sigma_1 \cdot \pi_1 \cdots \chi_n \cdot \sigma_n \cdot \pi_n \cdot y \longact{\pop(\gamma)} y' \cdot \mu
  $$
  where the runs $\pi_i$, $\chi_i$ and $\sigma_i$ satisfy the
  following conditions:
  \begin{enumerate}[(a)]
  \item\label{enum:decomp-only-p}
    $\pi_i$ consists of moves of process $p$ which are either
    local actions or
    sends that are unmatched in $\rho$,
  \item\label{enum:decomp-only-not-p}
   $\chi_i$ contains no move of process $p$,
  \item\label{enum:decomp-sync}
   $\sigma_i$ is a rendezvous involving $p$,
  \item\label{enum:decomp-wf}
 the transitions $x \longact{\push(\gamma)}
  x'$ and $y \longact{\pop(\gamma)} y'$ are matching stack actions (of
  process $p$), 
  \item\label{enum:schedule-p-first}
    for each $1 \leq i \leq n$, the run $\chi_i \cdot \sigma_i$ is not
    order-equivalent to a run of the form $\chi'_i \cdot \sigma'_i
    \cdot \chi''_i$ where $|\chi'_i| < |\chi_i|$ and $\sigma'_i$ is a
    rendezvous involving $p$.
  \end{enumerate}
  The scheduling of $p$'s actions as early as possible is expressed by
  condition~(\ref{enum:schedule-p-first}) (notice that $\s_i$ and
  $\s'_i$ correspond to the same send/receive pair). 
  
  \noindent
  We first show the following claim.

  \begin{claim}
    For each $1 \leq i \leq n$, all processes that move in $\chi_i$
    have an empty stack at the start and end of $\chi_i$. 
  \end{claim}

  To prove the claim, let us denote by $P_i = \{q_1, \ldots, q_k\}$
  the set of processes that move in $\chi_i$, ordered by their last
  occurrence in $\chi_i$. 
  Since the last action in $\chi_i$ is performed by $q_k$, 
  we derive from~(\ref{enum:schedule-p-first}) that the rendezvous
  $\sigma_i$ is on a channel between $p$ and $q_k$. 
  Now let $1 \leq h < k$.  It follows
  from~(\ref{enum:schedule-p-first}) that the last action of $q_h$ in
  $\chi_i$ is a communication action $b_h$.  We have two cases to
  consider: 
  \begin{iteMize}{$\bullet$}
  \item $b_h$ is a send action: If there was no matching receive in $\rho'$, then this send action could be scheduled after $\sigma_i$, contradicting~(\ref{enum:schedule-p-first}).
    Hence, $\rho'$ contains a matching receive, which, by eagerness, is the next action in $\rho'$.
    This matching receive is performed by a process $q_g$ with $h < g$.
  \item $b_h$ is a receive action: Since $\rho'$ is eager, the matching send is the previous action in $\rho'$.
    This matching send is performed by a process $q_g$.
    Moreover, we must have $h < g$ since, otherwise, this matched send/receive pair could be scheduled after $\sigma_i$, contradicting~(\ref{enum:schedule-p-first}).
  \end{iteMize}
  We obtain that, for every $1 \leq h < k$, the last action of $q_h$
  in $\chi_i$ is a communication action over a channel $c_h$
  satisfying $q_h\uchannel{}{c_h}{}q_g$ for some $h < g \le k$.
  Let $c_k$ denote the channel of the rendezvous $\sigma_i$, and recall that $q_k\uchannel{}{c_k}{}p$.
  Observe that $p$ is unrestricted on $c_k$ since, according to~(\ref{enum:decomp-wf}), the stack of $p$ is non-empty in $\sigma_i$.
  As the typed topology of $\Rqcp$ is non-\converging, we derive that
  $q_h$ is restricted on $c_h$ for each $1 \leq h \leq k$, since there
  is a simple undirected path $q_h \uchannel{}{c_h}{} \cdots \uchannel{}{}{} q_k
  \uchannel{}{c_k}{} p$ for each $h$. 
It follows that $q_h$ has an empty stack at the end of $\chi_i$.

  We have thus shown that, for each $1 \leq i \leq n$, all processes
  that move in $\chi_i$ have an empty stack at the end of $\chi_i$.
  Now, recall that $\rho'$ is well-formed since it is order-equivalent
  to $\rho$.  Therefore, all processes that move in $\chi_i$ also have
  an empty stack at the start of $\chi_i$, which concludes the proof
  of the claim.

  It follows from the claim that each run $\chi_i$ is well-formed, so
  $\mu$ is also well-formed.  Since the runs
  $\chi_i$ and $\mu$ are eager, we derive from the induction
  hypothesis that each $\chi_i$ is order-equivalent to a run $\chi'_i$
  that is both eager and well-bracketed, and, similarly, $\mu$ is
  order-equivalent to a run $\mu'$ that is both eager and
  well-bracketed.  Replacing in $\rho'$ each $\chi_i$ by $\chi'_i$ and
  $\mu$ by $\mu'$,  yields a run $\rho'' \sim
  \rho$ that is both eager and well-bracketed (the second condition
  for well-bracketed runs is satisfied since the runs $\chi_i$ contain
  no move of $p$).  This concludes the
  proof of the proposition.
\end{proof}

Well-bracketed runs in an (arbitrary) \rqcp cannot exploit the full
power of the multiple stacks.  Indeed, the well-bracketing property
ensures that the individual process stacks do not ``interact'' with
each other: a single, global stack is sufficient to simulate the run.
More precisely, given an \rqcp $\Rqcp =
\tuple{\Topo,\tau,M,(\Pdp^p)_{p\in P}}$, with
$\Pdp^p=(Z^p,z^p_\init,A^p,\Gamma^p,\Delta^p)$ for each $p \in P$, we
construct a product pushdown system $\Pdp^\otimes$ that simulates the
well-bracketed eager runs of $\Rqcp$.  Its set of control states is
$Z^\otimes = P \times (\prod_{p\in P} Z^p) \times \pow{P} \times
\pow{C}$.  A control state $(p, \vec{z}, E, G) \in Z^\otimes$ means
that $p$ is the active process, $\vec{z}$ is the current global
control state, $E$ is the set of processes that have an empty 
stack, and $G$ is the set of channels that are ``growing'', i.e., for
which no receive action is possible anymore.  The stack alphabet of
$\Pdp^\otimes$ is the disjoint union $\Gamma^\otimes = \bigcup_{p \in P} \Gamma^p$. The
stack of $\Pdp^\otimes$ will be the concatenation of $|P|$
words $u^p \in(\G^p)^*$, one for each process $p$, where $u^p$ is empty if
and only if $p \in E$.

Let us explain how the simulation of eager, well-bracketed runs works.
First, an active process $r$ is non-deterministically chosen, leading
to the control state $(r, (z^p_\init)_{p \in P}, P, \emptyset)$.
Then, $\Pdp^\otimes$ simulates the behavior of $r$ as expected, using
its stack as $r$ would do, but also updates the set $E$ accordingly. 
To simulate send actions $c!m$, $\Pdp^\otimes$ non-deterministically
decides whether $c!m$ is actually part of a rendezvous  on $c$
(provided that $c \not\in G$), or will never be matched. 
In the former case, $\Pdp^\otimes$ simulates (in a single step) the
rendezvous $c!m \cdot c?m$.
In the latter case, the channel $c$ is added to the set $G$ of
``growing'' channels. 
Moreover, in both cases, the communication is performed only if the
typed topology allows it, which can be checked using the set $E$. 

The pushdown system $\Pdp^\otimes$ may choose
non-deterministically, at any
time, to switch the active process to
some process $q$.  
Since the run simulated by $\Pdp^\otimes$ is well-bracketed, either
$q$'s stack is empty ($q \in E$) or 
the top stack symbol must belong to $\G^q$. Thus,  $\Pdp^\otimes$
performs this check and then sets  the active process to $q$.

By construction, the pushdown system $\Pdp^\otimes$ simulates all runs
of $\Rqcp$ that are both eager and well-bracketed, and only those
runs. Moreover, the size of $\Pdp^\otimes$ is bounded by
$|\Rqcp|^{\mathcal{O}(|P|\cdot |C|)}$. Since every \rqcp can be easily
modified in order to reach a given state with all stacks empty we
obtain:

\begin{proposition}\label{prop:confl-reduction}
  State eager-reachability of an \rqcp of size $n$ with
  non-\converging typed topology $\tuple{\Topo=(P,C),\tau}$ 
  reduces in \dexptime to state reachability for a pushdown system of
  size $n^{\mathcal{O}(|P|\cdot |C|)}$.
\end{proposition}

Since the state reachability problem for pushdown systems is decidable in
deterministic polynomial time, we obtain the upper bound:

\begin{proposition}\label{prop:eager-upperb}
  The state eager-reachability problem for \rqcp over a
  non-\converging typed topology  is in \dexptime.
\end{proposition}

\section{Eager \qcp and the Mutex Restriction}
\label{sect:mutex}

The previous section showed how to decide the state
eager-reachability problem provided that the topology behaves well
w.r.t.~pushdowns and communication. A
first natural question is whether one can decide if eager runs
suffice for solving the reachability problem. A second legitimate
question is whether the restriction to eager runs is realistic. We
answer to the first question negatively.  However, on the positive
side we show a restricted class of \qcp where eager runs suffice:
\qcp over cyclic topologies with the mutex restriction. We focus in
this section
on \qcp since the eager condition talks about communication only.

\begin{definition}
   A configuration $x$ of a \qcp $\Qcp$ is \emph{mutex} if for every simple
   undirected cycle $(p_0,c_1,p_1,\ldots,c_n,p_n=p_0)$ in the topology of
   $\Qcp$, at most one of the channels $c_i$ is non-empty in $x$.
   A run $\rho$ in $\Qcp$ is \emph{mutex} if each configuration in
   $\rho$ is mutex.
\end{definition}
A \qcp $\Qcp$ is called \emph{mutex} if every configuration reachable in
$\Qcp$ is mutex.
We show later in this section that the mutex property is decidable for
finite \qcp.
Notice also that every \qcp with polyforest topology is mutex.

\smallskip

Before discussing mutex we first comment on the results
of~\cite{latorre-s-2008-299-a} and explain their relation with
Theorem~\ref{thm:conv} and Corollary~\ref{c:mutex} below.
The latter paper shows that state reachability is
decidable for finite \qcp over polyforest topologies, and for
well-queueing \rqcp over directed forests.  The proof of the result
for \rqcp relies on the idea that, on tree topologies, one can
reorder runs such that the resulting run has a bounded
number of contexts, where in each context only
one process executes all its actions by reading on one unique
incoming channel from its tree parent (and---in the case of
\rqcp---solely when its local stack is empty). Hence, the
problem reduces to the control-state reachability for a bounded-phase
multi-stack pushdown system, a question which was proven to be decidable in
doubly exponential time~\cite{latorre-s-2007-161-a}. A simple channel
reversal argument allows us to reduce the question for finite \qcp over
polyforest topologies to directed forests.

We show in the following that mutex \qcp are eager.  This
allows us to apply the results of the previous section and to obtain
the decidability of state reachability (for both finite \qcp
over polyforest topologies and well-queueing \rqcp over directed
forests)  via a direct
proof. Moreover, recall that the complexity of the algorithm of the
previous section is \dexptime, so one exponential less than the
results obtained  in~\cite{latorre-s-2007-161-a} for polyforest architectures.

\begin{remark}
  Over a topology of two finite
  processes connected by two channels, mutex runs are referred to as
  ``half-duplex communication''.
  For these, it is known how to decide the
  reachability problem through an effective construction of the
  recognizable reachability set~\cite{cece-g-2005-166-a}.
  Quasi-stable  systems are a semantic ad-hoc extension of
  this idea to finite \qcp with larger, cyclic
  topologies~\cite{cece-g-1997-304-a}.
\end{remark}

\begin{proposition}
  \label{p:one-bound}
  Given a \qcp $\Qcp$, every mutex run starting with empty channels
  admits an order-equivalent eager run.
\end{proposition}
\begin{proof}
  By induction on the length of runs.
  The basis is trivial.
  Consider a mutex run $\rho$ of non-zero length, that starts with
  empty channels. In particular, each receive action in
  $\rho$ has a matching send in $\rho$.
  We write $P_\rho \subseteq P$ for the (non-empty) set of all processes
  $p$ that move in $\rho$.
  For each $p \in P_\rho$, let $e_p$ denote the last action of $p$ in
  $\rho$.
  If some $e_p$ is a local action, or a send action that is not matched in
  $\rho$, we may schedule it last, which preserves the run's mutex property,
  and derive the existence of an eager run $\rho' \sim \rho$ by induction.
  Otherwise, for each $p \in P_\rho$, the action $e_p$ is a communication
  action that is matched in $\rho$, and we let $c_p$ denote the channel of
  $e_p$.
  Note that each $c_p$, for $p \in P_\rho$, is a channel between $p$ and
  another process in $P_\rho$, which we call its last peer.
  We may build an infinite sequence of processes in $P_\rho$ by picking an
  arbitrary process in $P_\rho$ and iteratively moving to its last peer.
  By the pigeonhole principle, there exist $p_0, \ldots, p_n$ in $P_\rho$,
  with $n > 0$, such that $(p_0, c_{p_0}, \ldots, p_n, c_{p_n}, p_0)$ is
  an undirected path in $\Topo$ and $p_0, \ldots, p_n$ are distinct.
  Moreover, we may assume w.l.o.g. that $p_0$ is the process that moves
  last in $\rho$ among $\{p_0, \ldots, p_n\}$.
  To simplify notation, let us simply write $e_i$ in place of $e_{p_i}$,
  and $c_i$ in place of $c_{p_i}$.
  Remark that the undirected path $(p_0, c_0, \ldots, p_n, c_n, p_0)$ must
  be a simple undirected cycle if $c_0 \neq c_1$.

  Let us show that $e_1, e_0$ is a pair of matching send/receive actions.
  Since $p_0\uchannel{}{c_0}{}p_1$ and $p_1$ stops moving before $p_0$ in
  $\rho$, the communication action $e_0$, which is matched in $\rho$,
  must be a receive action $e_0 = c_0 ? m_0$.
  We obtain that $\rho$ is of the form:
  $$
  \rho \ = \ \chi \cdot x' \longact{e_1} y' \cdot \chi' \cdot
  x'' \longact{c_0 ? m_0} y'' \cdot \chi''
  $$
  with no move of $p_1$ in $\chi'$, and no move of $p_0, p_1$ in $\chi''$.
  It follows that $c_0$ is non-empty in $y'$.
  Since $\rho$ is a mutex run, $x'$ and $y'$ are mutex configurations.
  If $c_0 \neq c_1$, then $c_0$ is also non-empty in $x'$, hence $c_1$
  must be empty in both $x'$ and $y'$, which is impossible since $e_1$
  is communication action on $c_1$.
  Therefore, we get that $c_0 = c_1$, and, hence, $e_1$ is the last send
  action on $c_0$ in $\rho$.
  Since $e_1$ is matched in $\rho$, it follows that $e_1$ is the matching
  send of $e_0$, which implies that $e_1 = c_0 ! m_0$.

  We may now conclude the proof of the proposition.
  Recall that $e_1, e_0$ are the last actions of $p_1$ and $p_0$ in
  $\rho$, respectively.
  Since $e_1 = c_0 ! m_0$ and $e_0 = c_0 ? m_0$ are matched,
  we may schedule $e_1, e_0$ last.
  This leads to a run $\rho'$ that is order-equivalent to $\rho$, and of
  the form:
  $$
  \rho' \ = \ \chi \cdot \mu \cdot x_0 \longact{c_0 ! m_0} x_1 \cdot
  \longact{c_0 ? m_0} \cdot x_2
  $$
  where the trace
  of $\mu$ satisfies
  $\trace(\mu) = \trace(\chi') \cdot \trace (\chi'')$.
  It follows from the previous trace equality that, for each configuration
  $(\vec{s}, \vec{w})$ occurring in $\mu$, there exists
  \begin{iteMize}{$\bullet$}
  \item
    either a configuration $(\vec{s'}, \vec{w'})$ in $\chi''$ with
    $\vec{w} = \vec{w'}$,
  \item
    or a configuration $(\vec{s'}, \vec{w'})$ in $\chi'$ such that
    $w'^{c_0} = w^{c_0} \cdot m_0$ and $w'^c = w^c$ for all $c \neq c_0$.
  \end{iteMize}
  In both cases, we derive that $(\vec{s}, \vec{w})$ is mutex since
  $(\vec{s'}, \vec{w'})$ is mutex.
  Therefore, the run $\mu$ is mutex.
  Moreover, the run $\chi$ is also mutex since it is a prefix of the mutex
  run $\rho$.
  We derive from the induction hypothesis that $\chi \cdot \mu$ is
  order-equivalent to an eager run $\mu'$.
  Replacing $\chi \cdot \mu$ by $\mu'$ in $\rho'$ yields a run
  $\rho'' \sim \rho$ that is eager.
  This concludes the proof of the proposition.
\end{proof}

\begin{corollary} \label{c:mutex}
  Every mutex \qcp is eager.
\end{corollary}

\begin{remark}
  A closer look at the proof of Proposition~\ref{p:one-bound} shows
  that the result still
  holds for the following weaker variant of the mutex property:
  a configuration $x$ of a \qcp $\Qcp$ is \emph{weakly mutex} if for
  every simple undirected cycle $(p_0,c_1,p_1,\ldots,c_n,p_n)$ in the
  topology of $\Qcp$, at most one of the channels $c_1, c_2$ is
  non-empty in $x$.
\end{remark}

We derive the following result as an immediate consequence of
Corollary~\ref{c:mutex}. The upper bound is obtained as an on-the-fly
simulation: since we simulate eager runs we do not have to store any
message, but  keep track of growing
channels. The lower bound follows from
the non-emptiness test of the intersection of several regular languages.

\begin{proposition} \label{thm:mutextoeager}
  The state reachability problem for finite, mutex \qcp
  is \pspace-complete.
\end{proposition}

\begin{remark}
  State reachability remains decidable for particular
  \emph{infinite-state} mutex \qcp. For example, if each local
  LTS  is a Petri net (i.e., the \qcp in question
  is a \fifo net~\cite{finkel-a-1987-106-a}), then the state
  reachability problem reduces to the Petri net reachability
  problem, which is known to be
  decidable~\cite{mayr-e-1984-441-a,kosaraju-s-1982-267-a}.
  \label{rem:qcp_local_dec}
\end{remark}

We end this section by showing that, for finite \qcp, the mutex
property is decidable (unlike the eager one).

\begin{proposition} \label{thm:testmutex}
  The question whether a finite \qcp is mutex, is \pspace-complete.
\end{proposition}
\begin{proof}
  Assume that $\Qcp$ is not mutex and consider a run $\rho$
  of minimal length from $x_\init$ to a configuration
  $x$ that is not mutex.
  By minimality, all configurations in $\rho$ up to $x$
  are mutex.
  Let $x'$ be the predecessor of $x$ in $\rho$.

  By Proposition~\ref{p:one-bound} we can reach $x'$ by an eager
  run $\rho'$ (which is generated on-the-fly in \pspace) and test whether
  there exists in $\Qcp$ a transition $x' \act{c!m} x$ that
  violates the mutex condition for $x$.  We guess $\rho'$ in \pspace
  (see remark above) and check
  whether there exists a simple undirected cycle $(p_0,c_1,p_1,\ldots,c_n,p_n)$
  in the topology of $\Qcp$ such that one channel $c_i$ is non-empty in $x'$
  and the action $c!m$ would write on another channel of this cycle
  (i.e., $c = c_j$ for some $j \neq i$).

  \pspace-hardness follows, again, by reducing from the non-emptiness
  test of the intersection of several regular languages.
\end{proof}

\begin{proposition}
  The question whether a finite \qcp is eager, is undecidable.
\end{proposition}
\begin{proof}
  We show a reduction from the universality problem for rational
  relations~\cite{berstel79}. Given such a relation
  $K \subseteq A^* \times B^*$, we ask whether $K = A^* \times B^*$.
  Here, $K$ is described by a finite automaton $\Aa_K$ over the alphabet
  $A \cup B$.

  We describe a finite \qcp over four processes, called $p_0,\ldots,p_3$,
  and four channels $c_{01}, c_{10}, c_{12}, c_{13}$ satisfying
  $p_0\dchannel{}{c_{01}}{}p_1$,
  $p_1\dchannel{}{c_{10}}{}p_0$,
  $p_0\dchannel{}{c_{02}}{}p_2$,
  $p_0\dchannel{}{c_{03}}{}p_3$.
  Process $p_0$ is described in Fig.~\ref{fig:P0}.
  The ingoing (outgoing, resp.) edges of $\Aa_K$ lead to the
  initial state (from the final states, resp.). Transition labels
  $a \in A$ in $\Aa_K$ are replaced by $c_{02} ! a$, and labels
  $b \in B$ are replaced by $c_{03} ! b$.

  Process $p_1$ is described in Fig.~\ref{fig:P1}.
  The LTS $\Lts^{p_2} = \Lts^{p_3}$ of processes $p_2,p_3$ consist of a
  single (initial) state without any transition.
  Therefore, when talking about ``state components'' below we only mention
  processes $p_0,p_1$.

\begin{figure}
  \centering
\begin{tikzpicture}[node distance=3cm, >=stealth', bend angle=30]
  \tikzstyle{every state} = [draw=gray, thick, fill=gray!20, minimum size=4mm, inner sep=2pt]

  \node[state,initial] (q_0)         {$0$};
  \node[state]	(q_1) [right of=q_0] {$1$};
  \node[state]	(q_2) [right of=q_1] {$2$};
  \node		(q_5) [rectangle,draw,below of=q_1,node distance=2cm,inner sep=5pt] {$\Aa_K$};

  \path[->] (q_0) edge              node [below]       {$c_{01} ! \mathtt{\$}$}     (q_1)
                  edge [bend right] node [below left]  {$\varepsilon$}                     (q_5)
            (q_1) edge              node [below]       {$c_{10} ? \mathtt{\$}$}     (q_2)
                  edge [loop above, min distance=7mm, in=60, out=120] node [above]       {$c_{02} ! a, \, c_{03} ! b$} (q_1)
            (q_5) edge [bend right] node [below right] {$\varepsilon$}                     (q_2);
\end{tikzpicture}
  \caption{Process $p_0$ ($a \in A, b\in B$)}
  \label{fig:P0}
\end{figure}
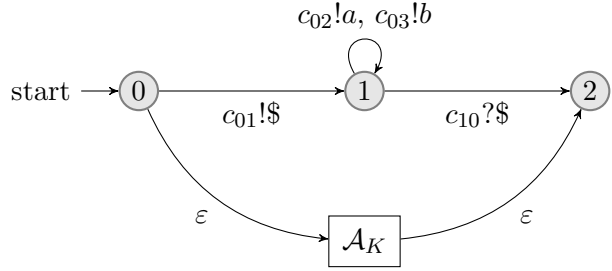

\begin{figure}
  \centering
\begin{tikzpicture}[node distance=3cm, >=stealth', bend angle=40]
  \tikzstyle{every state} = [draw=gray, thick, fill=gray!20, minimum size=4mm, inner sep=2pt]

  \node[state,initial] (q_0)	{$3$};
  \node[state]	(q_1) [right of=q_0] {$4$};
  \node[state]	(q_2) [right of=q_1] {$5$};

  \path[->] (q_0) edge              node [above] {$c_{10} ! \mathtt{\$}$} (q_1)
                  edge [bend right] node [below] {$\varepsilon$}                 (q_2)
            (q_1) edge              node [above] {$c_{01} ? \mathtt{\$}$} (q_2);
\end{tikzpicture}
  \caption{Process $p_1$}
  \label{fig:P1}
\end{figure}
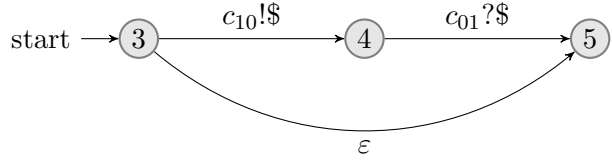

  The only runs of the above \qcp that cannot be reordered into an eager
  run are produced by $p_0$ and $p_1$ using all four
  $\mathtt{\$}$-transitions.  The state component of these
  configurations is $(2,5)$. The channel contents are
  $\varepsilon$ for $c_{01}$ and $c_{10}$, $A^*$ for $c_{02}$ and $B^*$ for
  $c_{03}$. Each of these configurations can be also reached by an
  eager run if and only if $K = A^* \times B^*$.
\end{proof}

\section{Bounded Phase Reachability}
\label{sect:bounded}

Bounded-context reachability has shown to be a successful
under-approximation method for the analysis of concurrent
Boolean programs~\cite{qadeer-s-2005-93-a}. For \rqcp, bounded-context
reachability allows us to attack the reachability problem from a
different angle than in Section~\ref{sect:converge}.  In this section,
we neither restrict the typed topology, nor constrain the runs to be
eager (or mutex). The price to pay is a (strong) restriction on the
form of the possible runs, namely a bounded number of switches between
processes (i.e., phases). Our construction
subsumes the \twodexptime algorithm for bounded-context
reachability of well-queueing recursive communicating processes, as
described in~\cite{latorre-s-2008-299-a}. Recall that the latter
algorithm is based on a reduction to bounded-phase reachability for
multi-stack systems. In contrast, our construction below is direct and
simpler.

\newcommand{\final}{f}
\newcommand{\Mux}{\textit{Mux}}
\newcommand{\Demux}{\textit{Demux}}

A \emph{phase} of an \rqcp is a run consisting of moves of a unique
process, called the \emph{phase process}. In order to get decidability
results one needs to introduce further restrictions over the
communications performed during a phase. The first, obvious,
restriction is on the typed topology $\tuple{\Topo,\tau}$: for every
channel $c$, either the source or the destination process is
restricted on $c$. Moreover, we assume for simplicity that for each
channel $c$, one of the two processes is unrestricted on $c$. The
second type of restriction concerns the kind of communication a
process is allowed to perform during a phase, and is defined by two
(dual) types of phases, called \emph{mux-phases} and
\emph{demux-phases}, respectively.

Let $c$ be a channel with \emph{source} process $p$ that is restricted
on $c$. A phase of process $p$ is a \emph{mux-phase} (with channel
$c$) if the allowed communication for $p$ is either sending into $c$,
or receiving on channels $d$ such that the source process is
restricted on $d$, see also Figure~\ref{fig:Mux-Demux}.  Dually, let
$c$ be a channel with \emph{destination} $p$ that is restricted on
$c$. A phase of process $p$ is a \emph{demux-phase} (with channel $c$)
if the allowed communication for $p$ is either receiving on $c$, or
sending on channels $d$ such that the destination process is restricted on
$d$. Demux-phases are precisely the phases/contexts used by
\cite{latorre-s-2008-299-a}.

\begin{figure}
  \centering
\begin{tikzpicture}[
  node distance=1.5cm,
  decoration={
    markings,
    mark=at position 0.62 with {\arrow{triangle 60}}
  }]

  \node (p)  {$p$};
  \node (p')  [right of=p]      {};
  \node (q1)  [above left of=p] {};
  \node (qn)  [below left of=p] {};

  \draw       (p) edge [o-] node [above, yshift=1mm] {$c$} (p')
  [decorate]  (p) to         (p')
  ;
  \draw      (q1) edge [o-*] (p)
  [decorate] (q1) to         (p)
  ;
  \draw      (qn) edge [o-*] (p)
  [decorate] (qn) to         (p)
  ;

  \draw [dotted, bend angle=30, bend right]
  ([yshift=-1mm] q1.south east) to ([yshift=1mm] qn.north east)
  ;
\end{tikzpicture}
\begin{tikzpicture}[node distance=0.5cm]
  \node (p)  {};
  \node (p') [right of=p, anchor=west]  {\small\em restricted};
  \node (q)  [below of=p] {};
  \node (q') [right of=q, anchor=west] {\small\em unrestricted};

  \draw (p) edge [o-] (p')
        (q) edge [*-] (q')
  ;
\end{tikzpicture}
\begin{tikzpicture}[
  node distance=1.5cm,
  decoration={
    markings,
    mark=at position 0.62 with {\arrow{triangle 60}}
  }]

  \node (p)  {$p$};
  \node (p')  [left of=p]        {};
  \node (q1)  [above right of=p] {};
  \node (qn)  [below right of=p] {};

  \draw       (p') edge [-o] node [above, yshift=1mm] {$c$} (p)
  [decorate]  (p') to       (p)
  ;
  \draw      (p) edge [*-o] (q1)
  [decorate] (p) to         (q1)
  ;
  \draw      (p) edge [*-o] (qn)
  [decorate] (p) to         (qn)
  ;

  \draw [dotted, bend angle=30, bend left]
  ([yshift=-1mm] q1.south west) to ([yshift=1mm] qn.north west)
  ;
\end{tikzpicture}
  \caption{Phases of an \rqcp: mux (on the left) and demux (on the right)}
  \label{fig:Mux-Demux}
\end{figure}
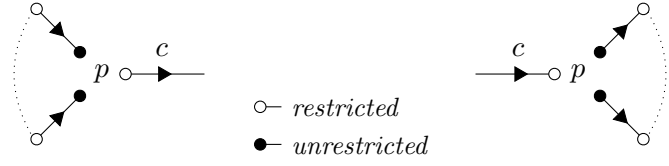

A run $\rho$ of an \rqcp is said to be \emph{$k$-bounded}, if it can
be decomposed as $\rho=\rho_1 \cdots \rho_k$ where each $\rho_j$ is a
mux- or demux-phase. A configuration $x\in X$ is
\emph{$k$-bounded-reachable} in a \rqcp $\Rqcp$ if there exists a
$k$-bounded run of $\Rqcp$ from the initial configuration $x_\init$ to
$x$.  We define the \emph{$k$-bounded-reachability set} of $\Rqcp$ as
$\Reach_k(\Rqcp)$, the set of $x \in X$ that are $k$-bounded-reachable in
$\Rqcp$.  The~\emph{state bounded-reachability problem} for \rqcp
asks for a given \rqcp $\Rqcp$, a global control state $\vec{z} \in
\vec{Z}$ and an integer $k$ (in unary encoding), whether
$\Reach_k(\Rqcp)$ intersects $\{\vec{z}\} \times (\prod_{p \in P}
(\Gamma^p)^*) \times (M^*)^C$.

In the remainder of this section we will use an extended version of
phases, still denoted as phase for convenience.
A~\emph{phase} $\phi=(p,\Pdp,z_F)$ will consist, as previously, of
a~\emph{phase process} $p \in P$ and a pushdown system
$\Pdp=(Z,z_I,A,A_\epsilon,\Gamma^p,\Delta)$ as in
Section~\ref{ssec:rqcp} (which may be, e.g., the
pushdown system of process $p$ in the \rqcp, up to changing the
initial state). In addition we specify a (control) state $z_F \in Z$,
which will be the target state of the phase. %
A phase is said to be \emph{local} if $A_\actcom$ is
empty.
The \emph{size} $|\phi|$ of a phase $\phi$ is
the number of control states of $\Pdp$.
We associate with a phase $\phi$ the binary
relation $\xrightarrow{\phi}$ over $(\prod_{p \in P} (\Gamma^p)^*)
\times (M^*)^C$, defined by
$(\vec{u}_I,\vec{v}_I)\xrightarrow{\phi}(\vec{u}_F,\vec{v}_F)$ if
there exists a run from the configuration
$(\vec{z}_I,\vec{u}_I,\vec{v}_I)$ to the configuration
$(\vec{z}_F,\vec{u}_F,\vec{v}_F)$ in the \rqcp obtained by fixing the
processes $q\not=p$ to the trivial pushdown system with one state and
no transition and the process $p$ to the pushdown system $\Pdp$. A
sequence $\Phi=(\phi_1,\ldots,\phi_k)$ of mux- or demux-phases
is called an~\emph{md-sequence}. Such a sequence is said
to be \emph{satisfiable} if the following relation holds:
$$((\varepsilon)_{p\in P},(\varepsilon)_{c\in
  C})\xrightarrow{\phi_1}\cdots\xrightarrow{\phi_k}((\varepsilon)_{p\in
  P},(\varepsilon)_{c\in C})$$
The \emph{size} of an md-sequence $\Phi=(\phi_1,\ldots,\phi_k)$ is
$|\Phi|=|\phi_1|+\cdots+|\phi_k|$.

We will decide the satisfiability of md-sequences by reducing
the problem to sequences of local phases. The reduction is
performed by replacing one by one (de)mux-phases by local
phases. We introduce a preorder over md-sequences, that will
decrease during the reduction. Let us first define the preorder
$\sqsubseteq$ over phases by letting $\phi\sqsubseteq
\psi$ if phases $\phi$ and $\psi$ have  the same phase process and the
communication actions of $\phi$ are included in the communication
actions of $\psi$. This preorder is extended component-wise over
md-sequences by letting $(\phi_1,\ldots,\phi_{k})\sqsubseteq
(\psi_1,\ldots,\psi_{k})$ if $\phi_j\sqsubseteq \psi_j$ for every
$j$.

\begin{proposition}\label{lem:red}
  Let $\Phi=(\phi_1,\ldots,\phi_k)$ be an md-sequence with at least
  one non-local phase. We can compute a finite set $F$ of
  md-sequences with $|F|\leq |\Phi|^k$ in time $\mathcal{O}(|F|)$ such
  that $\Phi$ is satisfiable if and only if $F$ contains a satisfiable
  md-sequence, and such that for every
  $\Psi=(\psi_1,\ldots,\psi_k)\in F$:
  \begin{iteMize}{$\bullet$}
  \item $|\Psi|\leq 2|\Phi|^2$
  \item $\Psi\sqsubseteq \Phi$ and there exists $j$ such that $\psi_j$
    is local whereas $\phi_j$ is not local.
  \end{iteMize}
\end{proposition}

\begin{proof}
  Since $\Phi$ contains at least one non-local phase, there exists a
  maximal index $j$ such that $\phi_j$ is demux non-local, or there
  exists a minimal index $j$ such that $\phi_j$ is mux non-local. We
  first explain why these two cases are symmetric. Given a phase
  $\phi=(p,\Pdp,z_F)$ where $\Pdp=(Z,z_I,A,A_\epsilon,\Gamma,\Delta)$,
  let $\bar{\phi}=(p,\bar{\Pdp},z_I)$ be the phase with
  $\bar{\Pdp}=(A,z_F,A,A_\epsilon,\bar{\Delta})$ the pushdown system
  obtained from $\Pdp$ by reversing the channels, exchanging push/pop
  actions and send/receive actions, and reversing the transition
  relation. We observe that
  $(\vec{u},\vec{v})\xrightarrow{\phi}(\vec{u}',\vec{v}')$ if and only
  if
  $(\vec{u}',\vec{v}')\xrightarrow{\bar{\phi}}(\vec{u},\vec{v})$. In
  particular $(\phi_1,\ldots,\phi_k)$ is satisfiable if and only if
  $(\bar{\phi}_k,\ldots,\bar{\phi}_1)$ is satisfiable. Since $\phi$ is
  a mux (resp.~demux) phase if and only if $\bar{\phi}$ is a demux
  (resp.~mux) phase, we obtain that the two cases above are
  symmetric. Thus, in the remainder of this proof we assume that there
  exists a maximal index $j$ such that $\phi_j$ is a non-local
  demux-phase.

  Let $\phi_j=(p,\Pdp,z_F)$ and
  $\Pdp=(Z,z_I,A,A_\epsilon,\Gamma,\Delta)$ be the pushdown system of
  $\phi_j$. Since $\phi_j$ is a demux-phase, messages are received
  from a unique channel, say $c$. Moreover, process $p$ is restricted
  on this channel.  Let us define the md-sequence $\Phi^\varepsilon$ from $\Phi$
  by removing communication actions in the $j$-th phase.

  In the sequel, we show how to build md-sequences
  $\Phi^\pi=(\phi_1^\pi,\ldots,\phi_k^\pi)$, where $\Phi^\pi$ is
  parametrized by a sequence $\pi=(z_r)_{s\leq r\leq j}$ of control
  states $z_r\in Z$ with $s<j$. Each sequence $\Phi^\pi$ is such that
  $\Phi^\pi\sqsubseteq \Phi$ with $\phi_j^\pi$ a local
  phase. In order to obtain a local
  phase $\phi_j^\pi$, i.e., a phase without any communication action,
  all communications with the pushdown system $\Pdp$ are simulated in
  the phases $\phi_s,\ldots,\phi_j$. Here, the integer $s$ is the
  index of the first phase that sends messages into channel $c$, that
  are received in the $j$-th phase. We show below that $\Phi$ is
  satisfiable if and only if $\Phi^\varepsilon$ or $\Phi^\pi$
  is satisfiable for some sequence $\pi$.%

 The state sequence $\pi=(z_r)_{s\leq
    r\leq j}$ provides checkpoints of the simulation of $\Pdp$ during the
  phases $\phi_s,\ldots,\phi_j$. In particular, states $z_r \in Z$
  in $\pi$ will be  
  assumed by process $p$ with empty stack, and the communication on
  channel $c$ during phase $r$ takes place between state $z_r$ and
  state $z_{r+1}$.

  Since $p$ is restricted on channel $c$, it
  receives messages from $c$ in the $j$-th phase with empty
  stack. Moreover, by the choice of $j$ and the fact that a
  satisfiable md-sequence must end with empty channels, process $p$
  sends no message during phase $j$ (otherwise, there would exist some
  demux, non-local phase after $j$, namely one where such messages
  would be received). By a well-known~\emph{saturation
    algorithm} we can compute in polynomial time (see for
    example~\cite{esparza}) from
  $\Pdp$ the set $R$ of pairs of control states $(z,z')\in Z\times Z$
  such that there exists an execution of $\Pdp$, consisting of stack
  actions and local actions only, from $(z,\epsilon)$ to $(z',\epsilon)$,
  (i.e., from empty stack to empty stack). Let
  $\phi_r=(q_r,\Pdp_r,t_{F,r})$ where
  $\Pdp_r=(T_r,t_{I,r},A,\Gamma,\Delta_r)$ with $s\leq r\leq j$.

  We first provide the definition of $\phi_r^\pi$ with $s<r<j$. Recall
  that $\pi=(z_r)_{s\leq r\leq j}$. The
  pushdown system $\Pdp_r^\pi$ is obtained by considering $|Z|$ many
  copies of $\Pdp_r$. Control states of these
  copies are identified by pairs $(t,z)\in T_r\times Z$. In these
  copies, actions that send messages to the channel $c$ are directly
  matched with actions that receive messages in $\Pdp$. More formally
  for every $(t,c!m,t')\in \Delta_r$ and $(z,c?m,z')\in\Delta$ we add
  a local action from $(t,z)$ to $(t',z')$. We also add transitions
  that simulate the effect of the stack of $\Pdp$. More precisely we
  add a local action from $(t,z)$ to $(t,z')$ for every $t\in T_r$ and
  for every $(z,z')\in R$. The initial state $t_{I,r}$ and the final
  state $t_{F,r}$ are replaced by $(t_{I,r},z_r)$ and
  $(t_{F,r},z_{r+1})$, resp.

  The definition of $\phi_s^\pi$ follows almost the same
  construction except that we should take into account the fact that
  in this phase we first perform moves that potentially send messages
  in $c$ and then non-deterministically we start to simulate the
  pushdown system $\Pdp$. The difference is due to the fact that
  some messages into channel $c$ can be received during some phase
  before the $j$-th one.  The simulation is performed with the
  construction presented in the previous paragraph. However we keep in
  $\Pdp_s^\pi$ the original pushdown system $\Pdp_s$ and we add a
  local action from $t$ to $(t,z_s)$ for every $t\in T_r$. The initial
  state $t_{I,s}$ is left unchanged and the final state $t_{s,F}$ is
  replaced by $(t_{s,F},z_{s+1})$.

  The definition of $\phi_j^\pi$ is obtained by a simpler
  construction. Since messages received from $c$ are simulated
  in the previous phases, we can remove the communication actions of
  $\Pdp$. Since the $j$-th phase may start or end with non-empty
  stack, we need in addition an extra copy of $\Pdp$  (also without communication actions). The copy of a
  control state $z$ is denoted by $\tilde{z}$. We then add a local action
  from $z_s$ to $\tilde{z}_j$ with the empty stack guard, i.e., this
  local action belongs to $A_\epsilon$. This action accounts for the
  simulation of $\Pdp$ between state $z_s$ and state $z_j$.
  Moreover, the initial control state $z_{I}$
  is left unchanged
and the final state is   replaced by $\tilde{z}_F$.

  Finally, the phases $\phi_r^\pi$ with $r<s$ or $r>j$ are equal
  to $\phi_r$. We observe that $\Phi$ is satisfiable if and only if
  $\Phi^\varepsilon$ is satisfiable or there exists a sequence $\pi$ such that $\Phi^\pi$ is
  satisfiable. Defining   $F$ as the set of md-sequences
  $\Phi^\pi$ and the additional md-sequence $\Phi^\varepsilon$ concludes the proof.
\end{proof}

\begin{corollary}\label{cor:sat}
  The satisfiability of an md-sequence $\Phi$ of length $k$ can be
  checked in time doubly exponential in $k$ (but polynomial in
 the size of  $\Phi$).
\end{corollary}

\begin{proof}
  Since the reduction introduced by applying Proposition~\ref{lem:red}
  transforms at least one non local phase into a local one,
  after at most $k$ steps we obtain a finite set $F$ of local
  phases. Moreover an immediate induction based on Proposition~\ref{lem:red}
  also shows that every $\Psi\in F$ has size $|\Psi|\leq
  2^k|\Phi|^{2^k}$. The size of $F$ can be bounded by the number
  of leaves of a tree of height $k$ with rank bounded by
  $(2^k|\Phi|^{2^k})^k$. Thus $|F|\leq ((2^k|\Phi|^{2^k})^k)^k$. The
  satisfiability of a sequence $\Psi\in F$ can be performed in time
  $\mathcal{O}(|\Psi|^2)$, since the empty stack control state reachability
  problem for pushdown systems is decidable in polynomial time. We
  conclude that the satisfiability of an md-sequence can be checked in
  \twodexptime, but polynomially in $|\Phi|$ when $k$ is fixed.
\end{proof}

\begin{theorem}\label{th:bounded}

The state bounded-reachability problem for \rqcp with typed topology
such that each channel is restricted at least at one extremity, is
\twodexptime-complete. If the number of phases and the typed topology are not part of the
input, the problem can be solved in polynomial time.
\end{theorem}

\begin{proof}
  For the upper bound we can assume w.l.o.g.~that we reach the target
  control state with all stacks and channels empty. For this, we
  can choose non-deterministically the push actions that will not be
  matched and, for each channel, the first message that will be no
  longer received. The bound follows then from Corollary
  \ref{cor:sat}.

  For the lower bound we can adapt proof ideas
  from~\cite{atig-f-2008-121-a,latorre-s-2008-33-a}, by showing how
  to simulate alternating Turing machines $M$ of exponential space by
  \rqcp with typed topology as in the statement of the theorem. If the
  space bound of $M$ is $2^k$ we use $\mathcal{O}(k)$ processes,
  called $p_0$ and $p_i, q^o_i,q^e_i$, $1 \le i \le k$. Process $p_0$
  is the only one using a stack, storing an accepting computation tree
  of $M$. We will not go into the details how to encode the tree (it
  is the usual depth-first traversal of the tree, plus appropriate
  encoding of transitions), see e.g.~\cite{atig-f-2008-121-a} for
  details. Instead we explain now how to check that the contents of
  the stack of $p_0$ is a word of the form $(w\#)^m$ for some $w \in \{0,1\}^{2^k}$ and $m >0$. 

In the first phase, process $p_0$ empties its stack and while doing
this, sends the following to $q^o_1,q^e_1$:
\begin{iteMize}{$\bullet$}
\item $q^o_1$: every symbol of $w$ at an odd position,
\item $q^e_1$: every symbol of $w$ at an even position.
\end{iteMize}
Assuming that the stack content of $p_0$ is $w_1 \# \cdots w_m \#$,
the outgoing channels of $p_0$ will contain after this first stage,
the following words ($u^o$ and $u^e$ denotes
the subword  of $u$ at odd and even positions, respectively):
\begin{iteMize}{$\bullet$}
\item $w_1^o \# \cdots \# w_m^o\#$  for $(p_0,q^o_1)$,
\item $w_1^e \# \cdots \# w_m^e\#$ for $(p_0,q^e_1)$.
\end{iteMize}
In the second and third phase, process $q^o_1$, and then $q^e_1$,
receives from $p_0$ and resends each message to $p_1$. In phases 4 and
5, process $p_1$ receives $w_1^o \# \cdots \# w_m^o\#$ from $q^o_1$,
and then $w_1^e \# \cdots \# w_m^e\#$ from $q^e_1$. In each of these
phases $p_1$ resends to $q^o_2$ and $q^e_2$ its odd/even subwords as
$p_0$ above, adding a separator \$ between the two halves. So process
$p_1$ acts basically like $p_0$, but on ``input'' of the form $w_1^o \# \cdots
\# w_m^o\# \$ w_1^e \# \cdots \# w_m^e\#$, where one has to check
equality for words of length $2^{k-1}$: $w_1^o=\cdots=w_m^o$ and
$w_1^e=\cdots=w_m^e$, respectively. This procedure is iterated up to
process $p_k$, that simply checks that it receives two words from
$q^o_k,q^e_k$ of the form
$((0\#0\# +1 \#1\#)\$^+)^*$.

The above proof for stack contents of the form $w\# w\# \cdots
    w\#$ for some $w \in \{0,1\}^{2^k}$, is of course a special
    case of the Turing machine simulation, however it captures the
    main idea. For the Turing machine it is readily seen how to extend
    the proof to a sequence of configurations $w_1\# w_2\# \cdots
    w_k\#$, where $w_{i+1}$ is the successor configuration of
    $w_i$. Here, it helps to see each $w_i$ as a sequence of 3 tape symbols,
    i.e., each position stores the current symbol, plus its
    neighbors. In addition, one encodes the transitions leading from
    $w_i$ to $w_{i+1}$, say after each $\#$. For the final check, 
    process $p_k$ will check  that the first triple is consistent with the
    middle symbol of the second triple.
\end{proof}

\section{Conclusion}

\noindent{\it Applications.\ }
\qcp combine an automata-based local process model with point-to-point
communication, which results
in an intuitive and simple framework.

Since we subsume well-queueing \rqcp, we also inherit their
application domains, e.g., event-based programs. The dual restriction
to well-queueing (i.e., that sending on a channel is only possible if
the stack is empty) covers, e.g., ``interrupt based'' programming
models, i.e., threads that can receive messages \emph{while} still in
recursion, as well as extended sensor networks where peers can
collect and send data \emph{while} using their pushdown for
computations.

Figure~\ref{fig:examplearch} shows an example for non-\converging typed
topologies that are on the rise with the current focus on distributed
computing. The topology corresponds to a hierarchical overlay network as
implemented, for example, in master-worker protocols.
Intuitively, each master distributes tasks to its workers and uses their
results during its own computation.
When the latter is finished, i.e., when its stack is empty, the master
sends a result to its own master.
Therefore, channel restrictions respect the hierarchy: channels between
a master and a worker must be restricted on the worker's side.
In fact, our generic non-\convergence criterion permits additional
communications: workers of the same master may communicate
with each other via channels on which they are restricted (e.g.,
$p_5$ and $p_6$),
and we may have a communication cycle between top-level masters
(e.g., $p_1$ and $p_2$).
Notice also the use of the dual notion to well-queueing, when
sending information from lower to higher
levels. %

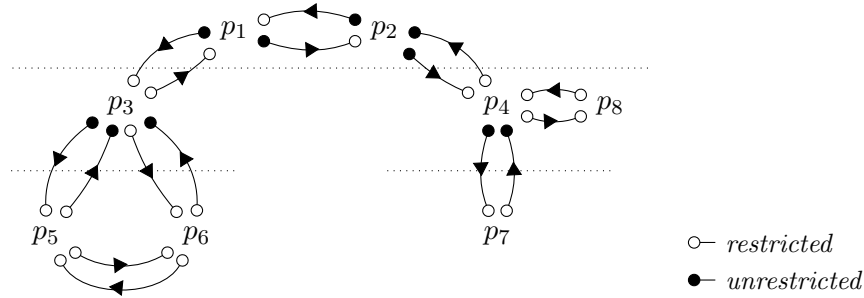
\begin{figure}[t]
  \centering
\begin{tikzpicture}[
  proc/.style = {}
  ]

   \node[proc] (l) {$p_3$};
   \draw (l)+(240:2) node[proc] (ll) {$p_5$};
   \draw (l)+(-60:2) node[proc] (lr) {$p_6$};

   \draw (l)+(1.5,1) node[proc] (t1) {$p_1$};
   \draw (l)+(3.5,1) node[proc] (t2) {$p_2$};
   \draw (l)+(5,0) node[proc] (r) {$p_4$};
   \draw ($(ll)!.5!(lr)$)+(5,0) node[proc] (rr) {$p_7$};
   \draw (r)+(1.5,0) node[proc] (r') {$p_8$};

   \begin{scope}[
     decoration={
       markings,
       mark=at position 0.56 with {\arrow{triangle 60}}
     }]

     \draw      (t2) edge[*-o,bend right=10] (r)
     [decorate] (t2) to  [    bend right=10] (r)
     ;
     \draw      (r) edge[*-o,bend right=20] (rr)
     [decorate] (r) to  [    bend right=20] (rr)
     ;
     \draw      (r) edge[o-o,bend right=20] (r')
     [decorate] (r) to  [    bend right=20] (r')
     ;
     \draw      (t1) edge[*-o,bend right=30] (l)
     [decorate] (t1) to  [    bend right=30] (l)
     ;
     \draw      (t2) edge[*-o,bend right=20] (t1)
     [decorate] (t2) to  [    bend right=20] (t1)
     ;
     \draw      (l) edge[*-o,bend right=30] (ll)
     [decorate] (l) to  [    bend right=30] (ll)
     ;
     \draw      (l) edge[o-o,bend right=10] (lr)
     [decorate] (l) to  [    bend right=10] (lr)
     ;
     \draw      (ll) edge[o-o,bend right=30] (lr)
     [decorate] (ll) to  [    bend right=30] (lr)
     ;
   \end{scope}

   \begin{scope}[
     decoration={
       markings,
       mark=at position 0.44 with {\arrowreversed{triangle 60}}
     }]

     \draw      (t2) edge[*-o,bend left=30] (r)
     [decorate] (t2) to  [    bend left=30] (r)
     ;
     \draw      (r) edge[*-o,bend left=20] (rr)
     [decorate] (r) to  [    bend left=20] (rr)
     ;
     \draw      (r) edge[o-o,bend left=20] (r')
     [decorate] (r) to  [    bend left=20] (r')
     ;
     \draw      (t1) edge[o-o,bend left=10] (l)
     [decorate] (t1) to  [    bend left=10] (l)
     ;
     \draw      (t2) edge[o-*,bend left=20] (t1)
     [decorate] (t2) to  [    bend left=20] (t1)
     ;
     \draw      (l) edge[*-o,bend left=10] (ll)
     [decorate] (l) to  [    bend left=10] (ll)
     ;
     \draw      (l) edge[*-o,bend left=30] (lr)
     [decorate] (l) to  [    bend left=30] (lr)
     ;
     \draw      (ll) edge[o-o,bend right=60] (lr)
     [decorate] (ll) to  [    bend right=60] (lr)
     ;
   \end{scope}

   \draw[dotted]
         (l)+(-1.5,0.5) coordinate (x)
         (r')+(0.5,0.5) -- (x);
   \draw[dotted]
         ($(ll)!.5!(lr)$) coordinate (x)
         ($(x)!.5!(l)$)+(-1.5,0) coordinate (y)
         (y)+(3,0) -- (y);
   \draw[dotted]
         ($(rr)!.5!(r)$)+(-1.5,0) coordinate (y)
         (y)+(3,0) -- (y);
\end{tikzpicture}
\begin{tikzpicture}[node distance=0.5cm]
  \node (p)  {};
  \node (p') [right of=p, anchor=west]  {\small\em restricted};
  \node (q)  [below of=p] {};
  \node (q') [right of=q, anchor=west] {\small\em unrestricted};

  \draw (p) edge [o-] (p')
        (q) edge [*-] (q')
  ;
\end{tikzpicture}
\caption{Non-\converging typed topology in a hierarchical
  master-worker setting}
\label{fig:examplearch}
\end{figure}

Proposition \ref{p:one-bound} allows for further
applications, since it does not assume that the \qcp is finite: we can
combine locally decidable models for
multi-threaded programs (with or without local data), as
well as local event-based programs together with eager (or mutex)
communication architectures; natural candidates for local
models would be Petri Nets, well-structured transition systems~\cite{finkel-a-2001-63-a}, or
multi-set pushdown systems~\cite{sen-k-2006-300-a}.

\medskip \noindent{\it Summary.\ } We discussed
in detail the class of eager \rqcp (as well as mutex \qcp)
which both generalize the current lineup of decidable models
for asynchronously communicating pushdown systems. Further,
we presented an optimal decision procedure for eager
\rqcp over non-\converging architectures in \dexptime,
as well as  a direct and simpler construction
for bounded phase reachability for \rqcp.

\medskip \noindent{\it Outlook.\ } This paper dealt with the most basic form of verification, namely control-state
reachability. More general reachability questions
(w.r.t.~configurations) may be interesting to consider.
Further decision problems for \qcp, like boundedness or liveness, will
be investigated in future work.

\bibliographystyle{abbrv}

\bibliography{bibliography}

\begin{thebibliography}{10}

\bibitem{abdulla-p-1996-91-a}
P.~A. Abdulla and B.~Jonsson.
\newblock Verifying programs with unreliable channels.
\newblock {\em Inf. Comput.}, 127(2):91--101, 1996.

\bibitem{atig-m-2010-117-b}
M.~F. Atig.
\newblock From multi to single stack automata.
\newblock In {\em Proc. of CONCUR 2010}, volume 6269 of {\em LNCS}, pages
  117--131. Springer, 2010.

\bibitem{atig-f-2008-121-a}
M.~F. Atig, B.~Bollig, and P.~Habermehl.
\newblock Emptiness of multi-pushdown automata is {2ETIME}-complete.
\newblock In {\em Proc. of DLT 2008}, volume 5257 of {\em LNCS}, pages
  121--133. Springer, 2008.

\bibitem{atig-m-2008-356-a}
M.~F. Atig, A.~Bouajjani, and T.~Touili.
\newblock On the reachability analysis of acyclic networks of pushdown systems.
\newblock In {\em Proc. of CONCUR 2008}, volume 5201 of {\em LNCS}, pages
  356--371. Springer, 2008.

\bibitem{berstel79}
J.~Berstel.
\newblock {\em Transductions and context-free languages}.
\newblock Teubner Studienb{\"u}cher, Stuttgart, 1979.

\bibitem{bouajjani-a-2005-348-a}
A.~Bouajjani, J.~Esparza, S.~Schwoon, and J.~Strejcek.
\newblock Reachability analysis of multithreaded software with asynchronous
  communication.
\newblock In {\em Proc. of FSTTCS 2005}, volume 3821 of {\em LNCS}, pages
  348--359. Springer, 2005.

\bibitem{bouajjani-a-2005-437-a}
A.~Bouajjani, M.~M{\"u}ller-Olm, and T.~Touili.
\newblock {R}egular symbolic analysis of dynamic networks of pushdown systems.
\newblock In {\em Proc. of CONCUR 2005}, volume 3653 of {\em LNCS}, pages
  473--487. Springer, 2005.

\bibitem{brand-d-1983-323-a}
D.~Brand and P.~Zafiropulo.
\newblock {O}n communicating finite-state machines.
\newblock {\em J. of ACM}, 30(2):323--342, 1983.

\bibitem{cece-g-1997-304-a}
G.~C{\'e}c{\'e} and A.~Finkel.
\newblock Programs with quasi-stable channels are effectively recognizable.
\newblock In {\em Proc. of CAV 1997}, volume 1254 of {\em LNCS}, pages
  304--315, 1997.

\bibitem{cece-g-2005-166-a}
G.~C{\'e}c{\'e} and A.~Finkel.
\newblock Verification of programs with half-duplex communication.
\newblock {\em Inf. Comput.}, 202(2):166--190, 2005.

\bibitem{esparza}
J.~Esparza, D.~Hansel, P.~Rossmanith, and S.~Schwoon.
\newblock Efficient algorithms for model checking pushdown systems.
\newblock In {\em Proc. of CAV 2000}, volume 1855 of {\em LNCS}, pages
  232--247. Springer, 2000.

\bibitem{esparza-j-2003-355-a}
J.~Esparza, A.~Kucera, and S.~Schwoon.
\newblock Model checking {LTL} with regular valuations for pushdown systems.
\newblock {\em Inf. Comput.}, 186(2):355--376, 2003.

\bibitem{finkel-a-1987-106-a}
A.~Finkel and L.~E. Rosier.
\newblock A survey on the decidability questions for classes of {FIFO} nets.
\newblock In {\em European Workshop on Applications and Theory of Petri Nets},
  volume 340 of {\em LNCS}, pages 106--132. Springer, 1987.

\bibitem{finkel-a-2001-63-a}
A.~Finkel and P.~Schnoebelen.
\newblock Well-structured transition systems everywhere!
\newblock {\em Theoretical Computer Science}, 256(1-2):63--92, 2001.

\bibitem{genest-b-2006-920-a}
B.~Genest, D.~Kuske, and A.~Muscholl.
\newblock A {K}leene theorem and model checking algorithms for existentially
  bounded communicating automata.
\newblock {\em Inf. Comput.}, 204(6):920--956, 2006.

\bibitem{genest-b-2007-1-a}
B.~Genest, D.~Kuske, and A.~Muscholl.
\newblock On communicating automata with bounded channels.
\newblock {\em Fundamenta Informaticae}, 80:147--167, 2007.

\bibitem{jhala-r-2007-339-a}
R.~Jhala and R.~Majumdar.
\newblock Interprocedural analysis of asynchronous programs.
\newblock In {\em Proc. of POPL 2007}, pages 339--350. ACM, 2007.

\bibitem{kidd10}
N.~Kidd, S.~Jagannathan, and J.~Vitek.
\newblock One stack to run them all.
\newblock In {\em Proc. of SPIN 2010}, volume 6349 of {\em LNCS}, pages
  245--261. Springer, 2010.

\bibitem{kosaraju-s-1982-267-a}
S.~R. Kosaraju.
\newblock Decidability of reachability in vector addition systems.
\newblock In {\em Proc. of STOC 1982}, pages 267--281. ACM, 1982.

\bibitem{latorre-s-2007-161-a}
S.~La~Torre, P.~Madhusudan, and G.~Parlato.
\newblock A robust class of context-sensitive languages.
\newblock In {\em Proc. of LICS 2007}, pages 161--170. IEEE Computer Society,
  2007.

\bibitem{latorre-s-2008-299-a}
S.~La~Torre, P.~Madhusudan, and G.~Parlato.
\newblock Context-bounded analysis of concurrent queue systems.
\newblock In {\em Proc. of TACAS 2008}, volume 4963 of {\em LNCS}, pages
  299--314. Springer, 2008.

\bibitem{latorre-s-2008-33-a}
S.~La~Torre, P.~Madhusudan, and G.~Parlato.
\newblock An infinite automaton characterization of double exponential time.
\newblock In {\em Proc. of CSL 2008}, volume 5213 of {\em LNCS}, pages 33--48.
  Springer, 2008.

\bibitem{lohrey-m-2004-160-a}
M.~Lohrey and A.~Muscholl.
\newblock Bounded {MSC} communication.
\newblock {\em Inf. Comput.}, 189(2):160--181, 2004.

\bibitem{mayr-e-1984-441-a}
E.~W. Mayr.
\newblock An algorithm for the general {P}etri net reachability problem.
\newblock {\em SIAM J. Comput.}, 13(3):441--460, 1984.

\bibitem{mayr}
R.~Mayr.
\newblock Process rewrite systems.
\newblock {\em Inf. Comput.}, 156(1-2):264--286, 2000.

\bibitem{papadimitriou-c-1994--a}
C.~Papadimitriou.
\newblock {\em Computational Complexity}.
\newblock Addison Wesley, 1994.

\bibitem{qadeer-s-2005-93-a}
S.~Qadeer and J.~Rehof.
\newblock Context-bounded model checking of concurrent software.
\newblock In {\em Proc. of TACAS 2005}, volume 3440 of {\em LNCS}, pages
  93--107. Springer, 2005.

\bibitem{ramalingam-g-2000-416-a}
G.~Ramalingam.
\newblock Context-sensitive synchronization-sensitive analysis is undecidable.
\newblock {\em ACM Trans. Program. Lang. Syst.}, 22(2):416--430, 2000.

\bibitem{sen-k-2006-300-a}
K.~Sen and M.~Viswanathan.
\newblock Model checking multithreaded programs with asynchronous atomic
  methods.
\newblock In {\em Proc. of CAV 2006}, volume 4414 of {\em LNCS}, pages
  300--314. Springer, 2006.

\bibitem{seth-a-2010-615-a}
A.~Seth.
\newblock Global reachability in bounded phase multi-stack pushdown systems.
\newblock In {\em Proc. of CAV 2010}, volume 6174 of {\em LNCS}, pages
  615--628. Springer, 2010.

\end{thebibliography}
\end{document}